\documentclass[a4paper, 11pt]{article}
\usepackage{mathrsfs}
\usepackage{graphicx}
\usepackage{amsmath}
\usepackage{amsthm}
\usepackage{amsfonts}

\usepackage{caption}
\usepackage{subcaption} %%both for subtables

\usepackage[english]{babel}

\usepackage{color}
\usepackage{enumerate}
\usepackage{bm}
\usepackage{url}
\usepackage{algorithm} %format of the algorithm
\usepackage{algorithmic} %format of the algorithm
\usepackage{mathtools}
\usepackage[mathcal]{euscript}

\usepackage{eufrak}

\DeclareFontFamily{OT1}{pzc}{}
\DeclareFontShape{OT1}{pzc}{m}{it}{<-> s * [1.200] pzcmi7t}{}
\DeclareMathAlphabet{\mathpzc}{OT1}{pzc}{m}{it}

\newcommand{\R}{\mathbb{R}}

\newcommand{\V}{\mathrm{V}}

\newcommand{\1}{\mathbf{1}}
\DeclareMathOperator{\diag}{diag}
\DeclareMathOperator{\rank}{rank}

\newtheorem{theorem}{Theorem}
\newtheorem{definition}{Definition}

\newtheorem{lemma}{Lemma}
\newtheorem{remark}{Remark}

\DeclareFontFamily{OT1}{pzc}{}
\DeclareFontShape{OT1}{pzc}{m}{it}{<-> s * [1.200] pzcmi7t}{}
\DeclareMathAlphabet{\mathpzc}{OT1}{pzc}{m}{it}

\newcommand*\mcapinn[2]{\vcenter{\hbox{$\mathsurround=0pt
  \ifx\displaystyle#1\textstyle\else#1\fi\bigcap$}}}

\newcommand*\mcupinn[2]{\vcenter{\hbox{$\mathsurround=0pt
  \ifx\displaystyle#1\textstyle\else#1\fi\bigcup$}}}

\def\begequarr{\begin{eqnarray}}
\def\endequarr{\end{eqnarray}}
\def\begequarrs{\begin{eqnarray*}}
\def\endequarrs{\end{eqnarray*}}
\def\begequ{\begin{equation}}
\def\endequ{\end{equation}}
\def\begequs{\begin{equation*}}
\def\endequs{\end{equation*}}
\def\begite{\begin{itemize}}
\def\endite{\end{itemize}}

\def\begcen{\begin{center}}
\def\endcen{\end{center}}
\def\begrem{\begin{remark}\rm}
\def\endrem{\end{remark}}
\def\ba{\begin{array}}
\def\ea{\end{array}}

\def\diag{\mbox{diag}}
\def\rank{\mbox{rank}\;}

\def\blkdiag{\mbox{blkdiag}\;}
\def\det{\mbox{det}\;}
\def\span{\mbox{span}\;}

\newcommand{\pdf}{\textnormal{pdf}}

\newcommand{\range}{\textnormal{range}}

\newcommand{\xb}{\mathbf{x}}
\newcommand{\yb}{\mathbf{y}}
\newcommand{\zb}{\mathbf{z}}

\newcommand{\Hb}{\mathbf{H}}

\newcommand{\Ib}{\mathbf{I}}

\newcommand{\Fb}{\mathbf{F}}

\newcommand{\Ab}{\mathbf{A}}

\newcommand{\Cb}{\mathbf{C}}
\newcommand{\Db}{\mathbf{D}}

\newcommand{\eb}{\mathbf{e}}

\newcommand{\db}{\mathbf{d}}

\newcommand{\phib}{\bm{\phi}}
\newcommand{\psib}{\bm{\psi}}

\newcommand{\mG}{\mathrm{G}}
\newcommand{\mV}{\mathrm{V}}

\newcommand{\gammab}{{\bm \gamma}}

\newcommand{\zetab}{\bm \zeta}

\newcommand{\betab}{\bm{\beta}}

\newcommand{\Exp}{\mathbb{E}}

\newcommand{\Eq}{\mathscr{E}}

\newcommand{\mcalQ}{\mathcal{Q}}

\def\beeq#1{\begin{equation}{#1}\end{equation}}
\title{\bf Differentially Private Distributed Computation \\ via Public-Private Communication Networks}
\date{}

\author{Lei~Wang,~Yang~Liu,~Ian~Manchester,~and~Guodong~Shi% <-this % stops a space
\thanks{L. Wang, I. Manchester and G. Shi are with Australian Centre for Field Robotics, The University of Sydney, Australia. E-mail: (lei.wang2; ian.manchester; guodong.shi@sydney.edu.au)}% <-this % stops a space
\thanks{Y. Liu is with Tencent Cloud Product Department, P.R. China. E-mail: (clarkfromthu@163.com)}% <-this % stops a space
}
\begin{document}
\maketitle

 \begin{abstract}
This paper studies the problem of multi-agent computation  under the differential privacy requirement of the agents' local datasets against eavesdroppers having access to node-to-node communications.
We first propose for the network equipped with public-private networks. The private network is sparse and not even necessarily connected, over which communications are encrypted and secure along with the intermediate node states; the public network is connected and may be dense, over which communications are allowed to be public. In this  setting, we propose a multi-gossip  Privacy-Preserving/Summation-Consistent (PPSC) mechanism over the private network, where at each  step, randomly selected node pairs update their states in such a way that they are shuffled with random noise while maintaining summation consistency. It is shown that this mechanism can achieve any desired differential privacy level with any prescribed probability. Next, we embed this mechanism in distributed computing processes, and propose privacy-guarantee protocols for three basic computation tasks, where an adaptive mechanism adjusts the amount of noise injected in PPSC steps for privacy protection, and the number of regular computation steps for accuracy guarantee. For average consensus, we develop a PPSC-Gossip averaging consensus algorithm by utilizing the multi-gossip  PPSC mechanism for privacy encryption before an averaging consensus algorithm over the public network for local computations.
For network linear equations and distributed convex optimization, we develop two respective distributed computing protocols by following the PPSC-Gossip averaging consensus algorithm with an additional projection or gradient descent step within each step of computation. Given any privacy and accuracy requirements, it is shown that all three proposed protocols can compute their corresponding problems with the desired computation accuracy, while achieving the desired differential privacy.
Numerical examples are used to illustrate the validity of the established theoretical results. In particular,  our framework demonstrates clear improvements in terms of learning accuracy for classification problems compared to  existing approaches under the same privacy budget.

 \end{abstract}

 \section{Introduction}

%%background and privacy risks of distributed computation
The study of distributed algorithms for network-wide computation problems over a multi-agent system is an emerging research topic of significance in the fields of smart grids \cite{berger2012smart},  mobile robotic networks \cite{sabattini2013decentralized}, intelligent
transportation \cite{Zhang2011} and machine learning \cite{boyd2011distributed,scaman2017optimal}.
In such problems, each agent (or node, representing a subsystem or a computing unit) over a network is assigned with a local dataset over which a local cost function is defined. Under a distributed computing scheme, the local datasets are encoded in either individual initial node states or update rules;  agents  share their node states over a communication network; these node states are updated based on the local datasets and the received neighboring states.  During such a computing process, when an attacker, a malicious user, or an eavesdropper has access to part or the entirety of node-to-node communications, the local  datasets/costs may be inferred. Since   the local dataset or  cost function of an agent    may  contain sensitive private information for the agent, new risks of  privacy breach arise.

%%existing privacy-preservation approaches and gap
In the literature, several insightful privacy-preserving distributed computing frameworks have been proposed. First of all, datasets themselves may be encrypted, e.g., via homomorphic encryption, and then used for optimization or computation, e.g.,  \cite{gentry2009fully,shoukry2016privacy,adria2017privacy}. For encryption-based approaches,  the resulting cyphertext is generally of high dimensionality, resulting in high communication and computation complexities. Besides, a quantitative privacy protection metric is usually absent.
%At a fundamental level, cryptography approaches are developed  for ensuring secure point-to-point communication \cite{Bellare2007Multi}, while eavesdroppers may be users participating in the computation process.
Another notable approach is to perturb the node communications or iterations with random noise \cite{huang2015differentially,hall2013differential},  providing a quantitative privacy guarantee and the robustness to post-processing and side information  under the notion of differential privacy \cite{dwork2006calibrating}.
The added random noises in the computation process, however, may significantly    jeopardize the computation accuracy \cite{huang2015differentially,nozari2016differentially,huang2012differentially,He2018Preserving}. There is a gap in the literature, where one cannot achieve  computational  accuracy, convergence efficiency, and  provable  privacy guarantees all together, but needs to search for convenient tradeoffs among the three aspects, for distributed computing.

We consider a multi-agent system where each agent holds a local private dataset, e.g., a local number,  equation, or function. The system seeks to compute a  value that   depends on   the  datasets at all agents. In distributed computing schemes,  agents hold individual dynamical states; share these states with neighbors over a communication network; and compute the update of the states based on their local datasets and received neighboring states. Specifically, we investigate  the following three distributed computation tasks that are extensively studied in the literature: average consensus  \cite{mo2016privacy,Nedic2009consensus}, network linear equations   \cite{shi2017networkTAC,Mou2015,Vempala2020}, and distributed convex optimization  \cite{nozari2016differentially,nedic2010constrained,konevcny2015federated}. We propose a public-private  setup for the node-to-node communication network. In this  public-private network, the public network  is connected and may be dense, over which communications are allowed to be public, while the private network is sparse and not even necessarily connected, over which communications are secure. We aim to develop distributed computing protocols such that the three considered computation tasks are
solved with any prescribed accuracy level, while the differential privacy is preserved with arbitrary privacy budgets, i.e., removing the
trade-off between computing accuracy and privacy.

%%%%%%%%%%%%%%%%%%%%%%%%%%

\subsection{Contributions}

%%Significance of Multi-Gossiping PPSC mechanism
 Over the private network, a Multi-Gossiping PPSC mechanism is proposed for injecting noises to node states to establish privacy protection, in a way that the summation of the node  stats is maintained. Given any privacy budgets, we prove that the Multi-Gossiping PPSC mechanism can achieve the desired differential privacy under an arbitrary probability.  The key idea for this strong privacy guarantee is enough amount of noise injected in the mechanism, and a sufficient number of recursions. In the meantime,  the mechanism itself maintains that the summation of the input and the output across the network stay consistent, a property that may be explored for accuracy satisfaction in distributed computation.

%%the proposed protocols and achieved results
Next, we embed the Multi-Gossiping PPSC in the distributed computing processes, and establish several privacy-guarantee protocols for the three considered  computation tasks. The following results are established.
 \begin{itemize}
  \item For  average consensus, we develop a PPSC-Gossip averaging consensus algorithm. The PPSC multi-gossiping mechanism is employed over the private network for privacy encryption, followed by an averaging consensus algorithm  over the public network for local computations. We prove an explicit  lower bound for the local computation depth, under which  the proposed algorithm can achieve the  differential privacy at any privacy level and with any prescribed  probability, while computing the exact network average at any accuracy level.

\item For network linear equations and distributed convex optimization, we develop two respective distributed computing protocols, where within each step of computation, the PPSC-Gossip averaging consensus algorithm is implemented with an additional  projection or gradient descent step. Given any prescribed levels for privacy and accuracy requirements, we prove that the two  protocols offer    both  privacy and   accuracy guarantees.
\end{itemize}

To the best of our knowledge, these results are the first kind in the literature with both proven differential privacy and computation accuracy guarantees. Of course, compared to existing privacy-preserving algorithms for these distributed computation problems, our privacy-preserving protocols   rely critically on secure communications over the private network and a generally more lengthy computation process. Also, we note that the private network is not necessarily connected, and thus it is not a straightforward conclusion that such a sparse private communication structure can deliver a global differential privacy assurance. In fact, a key idea in the protocols  lies in  an adaptive selection of the computation depth for accuracy guarantee,  according to the given privacy level.

\subsection{Related Work}

%The existing distributed computation;
Our work builds upon the existing literature of the differential privacy and three basic computation problems: average consensus, network linear algebraic equation and distributed convex optimization.  In the area of distributed computation, extensive research results have been reported to design and analyze algorithms for these three computation tasks both in continuous and discrete time as well as in deterministic and stochastic settings, see, e.g., \cite{Shi2015Consensus,shi2017networkTAC,Mou2015,Dekel2012Optimal,scaman2019optimal} and references therein.
Regarding the differential privacy, since the introduction in \cite{dwork2006calibrating}, it has gained a significant developments in many fields, including estimation \cite{Ny-Pappas-TAC-2014}, control \cite{kawano2020design}, learning  \cite{abadi2016deep}, etc.
Of most relevance to this paper are recent works \cite{huang2015differentially,nozari2016differentially,huang2012differentially,nozari2017differentially} on differentially private distributed computing algorithms. In \cite{huang2012differentially}, exponentially decaying noises are added to the communication messages for an  average consensus algorithm with a differential  privacy guarantee on initial values.
This idea is further developed in  \cite{nozari2017differentially,He2020TSP} for better noise adding mechanisms.
In \cite{huang2015differentially,nozari2016differentially},  distributed optimization problems with privacy-sensitive objective functions are addressed and differentially private computing protocols are developed by introducing random perturbations to  objective functions \cite{nozari2016differentially} or node states \cite{huang2015differentially}.
Similar ideas have been explored in \cite{han2016differentially} to preserve privacy of optimization constraints.
In all these results, the introduction of larger random noise can provide a better differential privacy guarantee on the one hand, but on the other hand leads to a larger computation error in the mean square sense.
That is, there is a trade-off between the differential privacy and computation accuracy.
It is worth pointing out that there are some computing algorithms, which can solve computing problems with any prescribed accuracy, but under other privacy notions. For example, in \cite{mo2016privacy}, by adding and subtracting decaying random noises to the averaging consensus process, it is shown that the average can be computed asymptotically in the mean square sense, while preserving the privacy in the sense that the maximum likelihood estimate of initial states has nonzero variance.

%Developments of gossip algorithms and its applications to privacy preservation.
The idea of PPSC protocol and its potential application to distributed optimization were reported in \cite{8619065}. This paper supersedes the work in \cite{8619065} by extending the PPSC framework to multi-gossiping setup, and propose and prove concrete privacy-preserving algorithms.
The proposed computing protocols are established on the multi-gossiping PPSC algorithm, which is an extension of the classic gossip process \cite{david2003gossip,Boyd2006Gossip}.  Instead of averaging the two  selected nodes by exchanging states \cite{david2003gossip,Boyd2006Gossip},  multiple node pairs are  selected at each time in our gossip algorithm and a directional communication of the perturbed state is conducted between each selected node pair, such that their states are shuffled with random noise while maintaining the summation.
Recent advances on gossiping protocols include new privacy-preserving gossip algorithms \cite{hanzely2017privacy}.
On the other hand, the idea of shuffling data for differential privacy is also studied in \cite{Cheu2019Distributed,Balcer2020Connecting}, where each agent randomizes its own local data, and then submits the resulting randomized data to a secure shuffler for random permutation before being public for computation purpose.  In such a protocol, a central shuffler is needed, and as in the previous differentially private computing protocols, it leads to a trade-off between the  privacy and  accuracy. As a comparison, our computing protocols are distributed  and there is no such a trade-off.

Our results are also related to the frameworks of federated learning \cite{konevcny2015federated},  a recent advance in privacy protection for machine learning.
In federated learning schemes, training datasets are distributed over a network of nodes, and a trusted center randomly selects a fraction of local nodes for model training, aggregates their local computations, and sends back the averaged updates as the network level decision.
When the communications between the computing center and decentralized nodes are accessed by an eavesdropper, the local training datasets may face privacy risks since information about such datasets is exposed in the communications \cite{kairouz2019advances}.
To handle such privacy concern, our PPSC framework may be embedded into the federated learning and provide enhanced privacy preservation as well, in the way that local nodes shuffle their decentralized computation results over a PPSC framework before sending them to the trusted center.

\medskip

\noindent{\bf Notation.} Denote $\eb_i$ as a basis vector whose entries are all zero except for the $i$-th being one.  We denote $\pdf(\cdot)$ as the probability density function of a random variable. For any matrix $\Ab\in\mathbb{R}^{m\times n}$, denote by $\sigma(\Ab)$ the set of all singular values of $\Ab$, and $\sigma_M(\Ab),\sigma_m(\Ab)$ the maximum and minimum singular values, respectively. For any set $X\subset\R^n$, we let $1_X(x)$ be a characteristic function, satisfying $1_X(x)=1$ for $x\in X$ and $1_X(x)=0$ for $x\notin X$.
Given any matrix, we denote its column  and row spaces by $\range(\cdot)$ and $\span(\cdot)$, respectively. With a slight abuse of notation, the range of a function is denoted by $\range(\cdot)$.
For any subspace $\Eq\in\R^n$, we denote $\mathpzc{P}_{\Eq}:\mathbb{R}^n\rightarrow\mathbb{R}^n$ as the orthogonal projection onto the subspace $\Eq$.
%For any number $a\in\R$, we use $\lfloor a \rfloor$ to denote the smallest integer that is smaller than or equal to $a$.

%There are three basic distributed computation tasks.

\section{Problem Definition}
\subsection{Distributed Computing for Multi-agent Systems}
We consider a multi-agent network with $n$ agents indexed in the set $\mathrm{V}=\{1,\dots,n\}$. The agents are equipped with   node-to-node communications and local private datasets. The goal of the system is to compute a network-level solution aggregated from all local datasets at each of the agents via distributed protocols.

\begin{itemize}
\item[(i)]  (Average consensus e.g., \cite{mo2016privacy,Nedic2009consensus}) Each agent $i$ holds a local private number $d_i\in\R$.
The computation goal of agents is to   compute the average $d^\ast=\frac{1}{n}\sum_{i=1}^{n}d_i$.

\item[(ii)] (Network linear equations  e.g., \cite{shi2017networkTAC,Mou2015,Vempala2020})  Each agent $i$ holds a private linear equation   $\mathfrak{E}_i: \,\Hb_i^\top\yb=z_i$ with an unknown $\yb\in\R^m$, $\Hb_i \in\mathbb{R}^{m}$ and $z_i\in \mathbb{R}$. The agents aim to solve  the overall linear equation
\begin{equation}\label{eq:LINEAR ALGEBRAIC EQUATION}
\mathfrak{E}: \quad \Hb\yb=\zb\,,
\end{equation}
with $\Hb\in\R^{n\times m}$ and $\zb\in\R^n$, the $i$-th row of which are $\Hb_i^\top$ and  $z_i$, respectively.

\item[(iii)] (Distributed convex optimization  e.g., \cite{nozari2016differentially,nedic2010constrained,konevcny2015federated})
Each agent $i$ holds a private convex function $f_i(\cdot):\R^m \rightarrow \R$. The computation goal of agents is to solve the  optimization problem
\beeq{\ba{rcl}\label{eq:opti}
\textnormal{minimize} && f(\yb):=\sum_{i=1}^n f_i(\yb)\,\\
\textnormal{subject to}&& \yb \in \mathtt{C}\,
\ea}
with a   compact convex set $\mathtt{C}\subset\R^m$.
\end{itemize}

\subsection{Public-Private Communication Networks}

%%%Define G_p
We propose for the multi-agent system to have a public-private network model: the public network $\mathrm{G}=(\mathrm{V}, \mathrm{E})$ and the private network  $\mathrm{G}_{\rm p}=(\mathrm{V}, \mathrm{E}_{\rm p})$, where the latter may or may not be a subgraph of the former (see Figure \ref{fig-GpG}).
We impose the following standing assumption throughout the paper.

\begin{figure}[h]
  \centering
  \includegraphics[width=5cm]{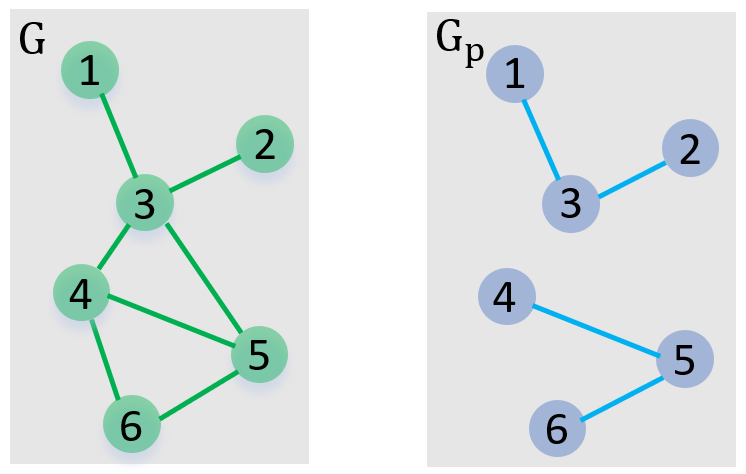}
  \caption{The public network ${\rm G}$  (left) and the private network ${\rm G}_{\rm p}$ (right). }
\label{fig-GpG}
\end{figure}

%\medskip

\noindent{\bf Standing Assumption.}
  (i) The public network  (graph) ${\rm G}$ is undirected and connected; Each node $i\in\mathrm{V}$ over  private network (graph) $\mathrm{G}_{\rm p}$ has a positive degree $r_i>0$. (ii) All node-to-node communications  over $\mathrm{G}_{\rm p}$ are secure and private; all node-to-node communicating messages over $\mathrm{G}$ are public.

\medskip

%%%Define Multi-Gossiping PPSC mechanism
Note that the private network $\mathrm{G}_{\rm p}$ may be very sparse, and it is \emph{not even necessarily connected}.  The  secure communications over the private network $\mathrm{G}_{\rm p}$ can be  realized via data encryption of communicating messages \cite{menezes1996handbook,Bellare2007Multi}.

 \subsection{Differentially Private Distributed Computation}

Denote $\mathscr{M}$ as the mapping that maps the collection of the  local datasets to all the node-to-node communications   over the public graph $\mathrm{G}$.
\medskip

\begin{definition}
(i) A distributed protocol is differentially private if  differential privacy for the mapping  $\mathscr{M}$ can be achieved  under any prescribed privacy budget.

(ii) A distributed protocol is computationally accurate if the protocol can output a solution that is within any prescribed error bound  at each agent of the system.
\end{definition}

We are interested in  distributed computing algorithms for the above three basic distributed computation tasks that offer both such differential privacy  and  computation accuracy.
\medskip

\noindent {\it Secure public network vs Public-private networks.} The proposed public-private network assumes a trusted  private network $\mathrm{G}_{\rm p}$.  If all node-to-node communications over the connected public graph $\mathrm{G}$ are made secure, then  standard distributed protocols will not face any privacy risk against communication eavesdroppers.
The introduction of   this public-private network setting has the following advantages:
\begin{itemize}
\item[(i)] Data encryption and decryption for establishing secure point-to-point communications are often computationally costly. Therefore, compared to having the entire network $\mathrm{G}$ equipped with secure node-to-node communications, a sparse  secure network $\mathrm{G}_{\rm p}$
provides improved scalability.
\item[(ii)] Even if all the links over the public  network $\mathrm{G}$ are made secure, it does not  preclude the possibility that certain node $k\in\mathrm{V}$ is a planted malicious node \cite{mo2016privacy}. Then all communications sent to this node $k$ and often the states of $k$'s neighbors become exposed for privacy risks despite the communications to $k$  being secure. While over the sparse private network $\mathrm{G}_{\rm p}$, each node only needs one secure neighbor. Therefore, the public-private network setting helps counter privacy risks caused by in-network malicious nodes.
\end{itemize}
The usage of a public-private network architecture has been proposed in the context of social computing, where the public social network is visible to everyone and a private social network is only visible to each node locally \cite{mirzasoleiman2016fast,chierichetti2015efficient}.

\section{PPSC Mechanism and PPSC-Gossip Averaging Consensus}
In this section, we   propose a multi-gossiping PPSC mechanism over the private network as a distributed random node states shuffling scheme with summation preservation. We prove that such a PPSC mechanism can be made  differentially private in terms of its input and output at an arbitrary privacy budget level in a probabilistic sense. Next, we show that the PPSC-Gossip mechanism over the private network can be combined with a standard consensus algorithm over the public network, and then an adaptive depth allocation will guarantee both differential privacy and computational accuracy.

%%%%%%%%%%%%%%%%%%%%%%%%%%

\subsection{Multi-Gossiping PPSC Mechanism}

Let  $\mathrm{G}_{\rm p}$ admit $q$ connected components with the $k$-th component denoted by   $\mathrm{G}_{\rm p}^k=(\mathrm{V}^k, \mathrm{E}_{\rm p}^k)$.
We propose the following  Multi-Gossiping Privacy-Preserving/Summation-Consistent (PPSC) mechanism, which consists of $S$ iterations  over the private network $\mathrm{G}_{\rm p}$.

\begin{algorithm}[H]\label{algorithm:PPSC}
\leftline{{\bf Input:} Initial states $\beta_i(0)$, $i\in\mV$.}

\leftline{{\bf For}  $t=1,\ldots,S$, {\bf run} the following iterations over $\mathrm{G}_{\rm p}$.}
\begin{itemize}
\vspace{-2mm}
  \item[1.] At each $\mathrm{G}_{\rm p}^k$, an agent $k_i$ is randomly selected with probability $1\over {n_k}$;
      an agent $k_j$ is then randomly selected from $k_i$'s neighbors with probability $1\over r_{k_i}$.
      Let $\mathrm{e}^k_{{\rm p}}(t)=(k_i,k_j)$ be the selected edge.
      \vspace{-3mm}
  \item[2.] Each agent $k_i$, $k=1,\ldots,q$ randomly and independently generates noise $\gamma_k(t)\sim \mathcal{N}(0,\sigma_\gamma^2)$, and sends $\omega_k(t) = \beta_{k_i}(t) - \gamma_k(t)$ to the agent $k_j$ over the edge $\mathrm{e}^k_{{\rm p}}(t)$.
      \vspace{-3mm}
  \item[3.] Each agent updates its state following
  \vspace{-3mm}
  			\[\ba{l}
            \beta_{k_i}(t) = \beta_{k_i}(t-1) - \omega_k(t)\,,\quad k=1,\ldots,q;\\
			\beta_{k_j}(t) = \beta_{k_j}(t-1) + \omega_k(t)\,,\quad k=1,\ldots,q;\\
            \beta_h(t) = \beta_h(t-1)\,,\quad\qquad\qquad h\in\mV\backslash \bigcup_{k=1}^q\{k_i,k_j\}\,.
			\ea\]
\end{itemize}
\leftline{{\bf Output:} $\beta_i(S)$, $i\in\mV$.}
\caption{Multi-Gossiping PPSC mechanism}
\vspace{-0mm}
\end{algorithm}

\begin{figure}[ht]
\centering
\begin{subfigure}{.2\textwidth}
\centering
		\includegraphics[width=.8\linewidth]{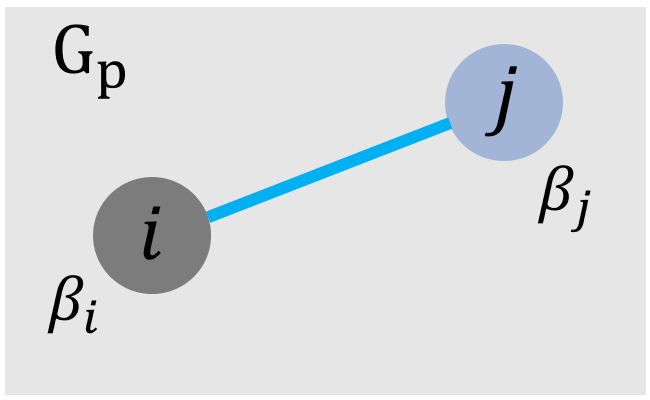}
\caption{ }
\label{fig-a}
\end{subfigure}
\begin{subfigure}{.2\textwidth}
\centering
		\includegraphics[width=.8\linewidth]{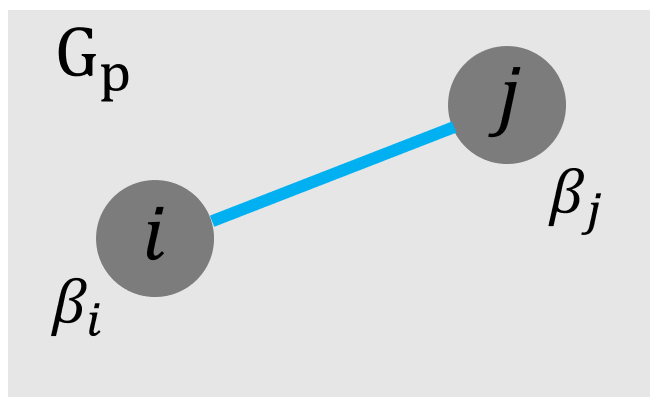}
\caption{ }
\label{fig-b}
\end{subfigure}
\begin{subfigure}{.2\textwidth}
\centering
		\includegraphics[width=.8\linewidth]{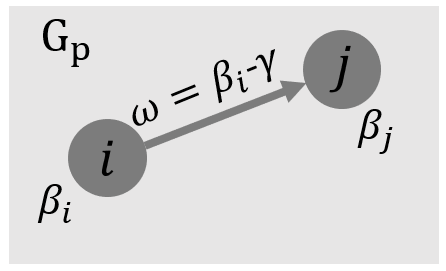}
\caption{ }
\label{fig-c}
\end{subfigure}
\begin{subfigure}{.2\textwidth}
\centering
		\includegraphics[width=.8\linewidth]{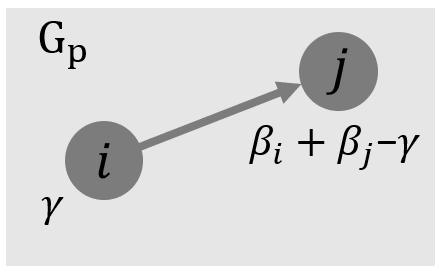}
\caption{ }
\label{fig-d}
\end{subfigure}
\caption{The gossip process between a pair of nodes. The selected agents and edge are highlighted in dark gray.
}
\label{fig-1}
\end{figure}

Note that in each iteration of the multi-gossiping PPSC mechanism, a pair of gossiping nodes essentially ``shuffle"  their states in a specific way: one node holds a noisy  summation of the two node states, and another node holds a noise correlated to the first node's new state.
This is related, but different from the idea of shuffling data    \cite{Cheu2019Distributed,Balcer2020Connecting}, where randomized local data is sent to a central secure shuffler,  and then the shuffler applies   random permutations before being public for the computation purpose. The  multi-gossiping PPSC mechanism does not rely on central secure shufflers; and the outcomes of the nodes states after the multi-gossiping PPSC mechanism and permutation shuffling are not the same.

%%Define DP of the $S$-step Multi-Gossiping PPSC mechanism
Taking a pair of nodes for instance, an overview of this gossip process is illustrated in Figure \ref{fig-1}.
The map from the input $\betab\in\mathbb{R}^n$  to the output $\mathbf{y}_{\rm ppsc}\in \mathbb{R}^n$ along the $S$-step Multi-Gossiping PPSC mechanism can be represented as $\mathscr{M}_{\rm S}$, i.e.,
$\mathbf{y}_{\rm ppsc}=\mathscr{M}_{\rm S}\big(\betab\big):=\mathsf{C}_{\mathsf{E}_{\rm p}} \betab + \mathsf{D}_{\mathsf{E}_{\rm p}} \gammab$,
where $\gammab=[\gammab_1;\ldots;\gammab_q]$ with $\gammab_k=[\gamma_k(1);\ldots;\gamma_k(S)]$, and $\mathsf{C}_{\mathsf{E}_{\rm p}}\in\R^{n\times n}$ and $\mathsf{D}_{\mathsf{E}_{\rm p}}\in\R^{n\times S}$ are some random matrices associated with the selected edge sequence $\mathsf{E}_{\rm p}= \big\{\mathsf{E}_{\rm p}^1,\ldots,\mathsf{E}_{\rm p}^q\big\}$ with $\mathsf{E}_{\rm p}^k=\big\{\mathsf{e}_{\rm p}^k(1),\ldots,\mathsf{e}_{\rm p}^k(S)\big\}$, $k=1,\ldots,q$. Let an eavesdropper with full observations of $(\mathsf{E}_{\rm p},\yb_{\rm ppsc})$  attempt to infer the input $\betab$. Under any observed $\mathsf{E}_{\rm p}$, we define two inputs $\betab^\prime, \betab\in\R^n$  as $\mu$-adjacent if they satisfy $\|\betab-\betab^\prime\|\leq \mu$ and $\mathsf{C}_{\mathsf{E}_{\rm p}}(\betab-\betab^\prime)\in \range\big(\mathsf{D}_{\mathsf{E}_{\rm p}}\big)$, and denote $\mathscr{M}_{\rm S}\big(\betab \,|\, \mathsf{E}_{\rm p}\big)$ as the  $\mathscr{M}_{\rm S}$ conditioned on the observed $\mathsf{E}_{\rm p}$.   We introduce the following differential privacy notion for the Multi-Gossiping PPSC mechanism \cite{dwork2006calibrating}.
\begin{definition}\label{def:DP}
Let  $\mu>0$, $\epsilon>0$ and $0<\delta<1/2$. The $\mathscr{M}_{\rm S}$ is termed to be $(\epsilon,\delta)$-differentially private  under $\mu$-adjacency if for all $R\subseteq \range(\mathscr{M}_{\rm S})$, there holds
  \beeq{\label{eq:def-DP}\ba{l}
  \mathbb{P}\left(\mathscr{M}_{\rm S}\big(\betab \,|\, \mathsf{E}_{\rm p}\big)\in R\right) \leq e^{\epsilon}\mathbb{P}\left(\mathscr{M}_{\rm S}\big(\betab^\prime \,|\, \mathsf{E}_{\rm p}\big)\in R\right) + \delta\,
  \ea}
  for any two $\mu$-adjacent $\betab,\betab^\prime\in\R^n$.
\end{definition}

%%Theorem
Denote $n_k=|\mathrm{V}^k|$, $n_{\max}=\max\{n_1,\ldots,n_q\}$, and $r^\dag = \min\{{r_{1}^{\rm min}}/{r_{1}^{\rm max}},\ldots,{r_{q}^{\rm min}}/{r_{q}^{\rm max}}\}$ with  $r_{k}^{\rm min}=\min\{r_{k_1},\ldots,r_{k_q}\}$ and $r_{k}^{\rm max}=\max\{r_{k_1},\ldots,r_{k_q}\}$, $k=1,\ldots,q$. Let $\kappa(\epsilon,\delta)=\frac{\mathcal{Q}^{-1}(\delta)+\sqrt{(\mathcal{Q}^{-1}(\delta))^2+2\epsilon}}{2\epsilon}$ with $\mathcal{Q}(w)=\frac{1}{\sqrt{2\pi}}\int_{w}^\infty e^{-\frac{v^2}{2}}d v$.
Denote by $\Theta_{\mathsf{E}_{\rm p}}$  the set of all possible edge sequences $\mathsf{E}_{\rm p}$ in the $S$-step Multi-Gossiping PPSC mechanism, and let $\lambda_{\rm ppsc} = \displaystyle\min_{\mathsf{E}_{\rm p}\in\Theta_{\mathsf{E}_{\rm p}}} \sigma_{m}^+\big(\mathsf{D}_{\mathsf{E}_{\rm p}}\big)$ with $\sigma_{m}^+(\cdot)$ denoting the minimal nonzero singular value. Finally, introduce $$
S_{\rho}^\ast = \Big({\log(1-\rho^{1\over q})-\log n_{\max}}\Big)/\Big({\log\big(1-\frac{1+r^\dag}{n_{\max}}\big)}\Big).
$$
Then we have the following result.

\begin{theorem}\label{theorem-PPSC}
Let $\mu>0$, $\epsilon>0$, $0<\delta<1/2$ and $0<\rho<1$. Suppose
$S \geq S_{\rho}^\ast$ and $\sigma_\gammab \geq {\mu \kappa(\epsilon,\delta)}/\lambda_{\rm ppsc}$.
Then  with a probability that is at least $\rho$, the $S$-Step Multi-Gossiping PPSC mechanism  is $(\epsilon,\delta)$-differentially private  under $\mu$-adjacency.
\end{theorem}

\subsection{Discussion: Privacy within the Private Graph}

Theorem \ref{theorem-PPSC} establishes the fact that the input-output mapping for an $S$-step PPSC multi-gossiping can be made differentially private with any privacy budget with an arbitrarily high probability. We have assumed that the node-to-node communications for the $S$-step PPSC multi-gossiping are secure in the standing assumption. Now, what if a malicious node, saying $m_\ast$, is within the multi-agent system, so that all communications associated with node $m_\ast$ become known to an eavesdropper (or equivalently, node $m_\ast$ is the eavesdropper)?

First of all, nodes never  send their true states to other nodes in the multi-gossiping PPSC procedure. Consequently, the presence  of such a malicious node does not impose {\em immediate} privacy concerns. Moreover,  since $\mathrm{G}_{\rm p}$ is not connected, only the nodes within the same connected component as $m_\ast$ are subjected to this {\em additional} privacy risk. Therefore, the sparsity  of $\mathrm{G}_{\rm p}$ becomes quite  useful.

Next,   consider the case where a  node $j$ is the sole neighbor of this   malicious node $m_\ast$ over $\mathrm{G}_{\rm p}$. Combing  the communication from node $j$, $\beta_j(0)-\gamma$, and the outcome of the PPSC mechanism $\gamma$ at node $j$, node $m_\ast$ will be able to infer the exact input  $\beta_j(0)$ of the node $j$. As a result, the privacy of node $j$ in terms of $\beta_j(0)$ will be fully lost to this malicious node $m_\ast$. To overcome this, it suffices for the node $j$ to have at least one  trustful  neighbor $k$, so that the state of node $j$ may have been shuffled between $j$ and $k$, which is not known to $m_\ast$. In general, if each node over $\mathrm{G}_{\rm p}$ has at least one trustful neighbor, the privacy of the nodes will have further guarantee  in terms of identifiability of $\beta_i(0),i\in \mathrm{V}$ from $\beta_i(S), i\in V$ in the presence of malicious nodes. Therefore, the structure of $\mathrm{G}_{\rm p}$ would enable stronger internal privacy preservation that goes beyond Theorem \ref{theorem-PPSC}, due to the shuffling effect \cite{Cheu2019Distributed,Balcer2020Connecting} that comes along the PPSC procedure. We leave a quantitive analysis for this PPSC enabled internal privacy protection  in future works since it is not fully aligning with the scope of the current paper.

\subsection{The PPSC-Gossip Averaging Consensus Algorithm}

%%Introduce PPSC-Gossip averaging consensus algorithm
Denote each {\em iteration} of the $S$-step Multi-Gossiping PPSC mechanism as $\mathsf{MultiGossipPPSC}$. Let each edge in $\rm G$ have the same weight $a\in(0,{1}/{n}]$, and denote ${\rm N}_i$ as the neighbor set of node $i\in{\rm V}$.
In the following,  the PPSC-Gossip averaging consensus (PPSC-Gossip-AC) algorithm is presented.
\begin{algorithm}[H]\label{algorithm:PPSC-consensus}
\leftline{{\bf Input:} The local private datasets $d_i$, $i\in\mV$, and parameters $S\in\mathbb{N}_+$ and $T\in\mathbb{N}_+$.}

\leftline{{\bf Initialize:} Set $s\gets0$ and $x_i(s)\gets  d_i$ for $i\in\mV$.}

\leftline{\mbox{$\quad$}{\bf For}  $s=1,\ldots,S$, over $\mathrm{G}_{\rm p}$ {\bf run}}

\leftline{\mbox{$\quad\qquad$} $\xb_s=\mathsf{MultiGossipPPSC}\big(\xb_{s-1}\big)$ with $\xb_s=[x_1(s);\ldots;x_n(s)]$}

\leftline{\mbox{$\quad$}{\bf For}  $s=S+1,\ldots,S+T$, over $\mathrm{G}$ {\bf run}}

\leftline{\mbox{$\quad\qquad$} $x_i(s) = x_i(s-1) + a\sum\limits_{j\in {\rm N}_i} (x_j(s-1)-x_i(s-1))$ for $i\in\mV$}

\leftline{{\bf Output:} $x_i(S+T)$, $i\in\mV$.}
\caption{PPSC-Gossip Averaging Consensus (PPSC-Gossip-AC) Algorithm }
\end{algorithm}

%%Insights of the algorithm
The PPSC-Gossip-AC algorithm is comprised of two stages: PPSC-Gossip stage and Average-Consensus stage. The former is to run  the multi-gossiping  PPSC mechanism with input $\xb_{0}$ and output $\xb_{S}$ over the private network $\mathrm{G}_{\rm p}$  for privacy encryption.
At the Average-Consensus stage, for $s=S+1,S+2,\ldots,S+T$, the standard averaging consensus algorithm is carried out over the public network $\mathrm{G}$ for local computations, where the node-to-node communications are $\xb_{S},\xb_{S+1},\ldots,\xb_{S+T-1}$.

Denote ${\bf d}=[d_1;d_2;\ldots;d_n]$ and $\chi_{\rm ea}=[\xb_{S};\xb_{S+1};\ldots;\xb_{S+T-1}]$, and define the map from ${\bf d}$ to $\chi_{\rm ea}$ as
$
\chi_{\rm ea} = \mathscr{M}_{\rm av}\big({\bf d}\big)\,.
$
Let an eavesdropper with full observations of the selected edges (i.e., $\mathsf{E}_{\rm p}$) during the PPSC-Gossip stage and the node-to-node communications (i.e., $\chi_{\rm ea}$) over the public graph $\rm G$.
In the following, similar to in  Definition \ref{def:DP}, we specify the differential privacy of the PPSC-Gossip-AC algorithm.
\begin{definition}\label{def:DP-av}
Let  $\mu>0$, $\epsilon>0$ and $0<\delta<1/2$. The $\mathscr{M}_{\rm av}$ is termed to be $(\epsilon,\delta)$-differentially private  under $\mu$-adjacency    if for all $R\subseteq \range(\mathscr{M}_{\rm av})$, there holds
  \beeq{\ba{l}
  \mathbb{P}\left(\mathscr{M}_{\rm av}\big({\bf d} \,|\, \mathsf{E}_{\rm p}\big)\in R\right) \leq e^{\epsilon}\mathbb{P}\left(\mathscr{M}_{\rm av}\big({\bf d}^\prime \,|\, \mathsf{E}_{\rm p}\big)\in R\right) + \delta
  \ea}
  for any two $\mu$-adjacent $\db, \db^\prime\in\R^n$.
\end{definition}

%Theorem
Denote by $\Ab\in\R^{n\times n}$ the Laplacian matrix of graph $\rm G$, and $\lambda_{{\rm G}}$ the algebraic connectivity  of ${\rm G}$, i.e., the second smallest eigenvalue of the Laplacian matrix $\Ab$. Clearly, $0\leq \lambda_{{\rm G}}<1$ by the standing assumption. One of our main results is summarized below.
\begin{theorem}\label{theorem-Ave}
For any $\mu>0$, $\epsilon>0$, $0<\delta<1/2$, $0<\rho<1$ and $\nu>0$, let $S \geq S_{\rho}^\ast$, $\sigma_\gammab \geq {\mu \kappa(\epsilon,\delta)}/\lambda_{\rm ppsc}$, and
\begin{equation}\label{eq:theo2-T}
T \geq \Big({\log\nu-\log(n\|\db\|^2 + 2q^2S^2\sigma_{\gammab}^2)}\Big)/\Big({2\log (1-\lambda_{\rm G})}\Big) \,.
\end{equation}
Then, the PPSC-Gossip-AC algorithm
  \begin{itemize}
    \item[(i)]preserves  $(\epsilon,\delta)$-differential privacy under $\mu$-adjacency with a probability that is at least $\rho$, and
    \item[(ii)]computes the average with a $\nu$-accuracy, i.e., $\mathbb{E}\,\|\xb_{S+T}- \frac{1}{n}{\bf 1}_n{\bf 1}_n^\top{\bf d}\|^2 \leq \nu \,.$
  \end{itemize}

\end{theorem}
By Markov's inequality, the condition $\mathbb{E}\,\|\xb_{S+T}- \frac{1}{n}{\bf 1}_n{\bf 1}_n^\top{\bf d}\|^2 \leq \nu $ guarantees that
$$
\mathbb{P} \Big( \|\xb_{S+T}- \frac{1}{n}{\bf 1}_n{\bf 1}_n^\top{\bf d}\| \geq \varsigma \Big)\leq \frac{\nu}{\varsigma^2}
$$
for any $\varsigma>0$. In other words, Theorem \ref{theorem-Ave} establishes that $\xb_{S+T}$ may get arbitrarily close to  $\frac{1}{n}{\bf 1}_n{\bf 1}_n^\top{\bf d}$ for arbitrarily high probability when $\nu$ is chosen to be small enough.

%\begin{remark}
%The $(\epsilon,\delta)$-differential privacy stated in Theorem \ref{theorem-Ave} indeed occurs with a probability $p$, which, termed as probabilistic $(\epsilon,\delta)$-differential privacy, is a novel extension of the differential privacy metric, to the best of our knowledge. We note that this probabilistic $(\epsilon,\delta)$-differential privacy is meaningful in practice since the most of events is practically unobservable. Besides, it is worth noting that this type of differential privacy is different from the  $(\epsilon,\delta)$-probabilistic differential privacy \cite{machanavajjhala2008privacy}, which indicates that the $\epsilon$- differential privacy occurs with a probability larger than $1-\delta$.
%
%\end{remark}

%
%\begin{remark}
%  We remark that there might be some eavesdroppers monitoring node-to-node communications only. In this case, it is highlighted that the same differential privacy can be realized even if the eavesdroppers have all node-to-node communications of the PPSC-Gossip stage.
%\end{remark}

%%%%%%%%%
\section{PPSC-Gossip Linear-Equation Solver}
\label{sec-eq}

In this section, the PPSC-Gossip-AC algorithm is explored to solve the network linear equation (\ref{eq:LINEAR ALGEBRAIC EQUATION}). Recall that in such a network linear equation, each agent $i$ holds   a linear algebraic equation $\mathfrak{E}_i: \,\Hb_i^\top\yb=z_i$ with an unknown $\yb\in\R^m$, $\Hb_i \in\mathbb{R}^{m}$ and $z_i\in \mathbb{R}$, and all agents aim to solve  the overall linear equation $\Hb\yb=\zb$, where   the $i$-th row of $\Hb\in\R^{n\times m}$ and $\zb\in\R^n$ are $\Hb_i^\top$ and  $z_i$, respectively.

Regarding solutions of the equation (\ref{eq:LINEAR ALGEBRAIC EQUATION}), there are three cases: (i) a unique exact solution; (ii) a unique least-squares solution; (iii) infinite number of solutions. In the following, we focus on the case (i), and suppose $\rank(\Hb)=m$ and $\zb\in\span(\Hb)$, which guarantees the unique exact solution $\yb^\ast:=(\Hb^\top\Hb)^{-1}\Hb^\top\zb$. We note that the case (ii) can be handled by adapting the algorithm in the next section.

%%%
\subsection{Adjacency of Linear Equations}

%%Adjacency of Linear Equations is specified by the distance of affine subspaces
Note that solutions of each linear equation $\mathfrak{E}_i$ specify a unique affine solution subspace $\Eq_i:=\big\{\yb\in\R^m:\Hb_i^\top\yb=z_i\big\}$. In the following, the adjacency of two linear equations is characterized by the distance of their corresponding affine solution subspaces.

%%Define two distances
For any two affine subspaces $\Eq_i:=\big\{\yb\in\R^m:\Hb_i^\top\yb=z_i\big\}$ and $\Eq_i^\prime:=\big\{\yb\in\R^m:{\Hb_i^\prime}^\top\yb=z_i^\prime\big\}$, we define two distances between them by
\begin{equation}\label{eq:r_d}
  \mbox{d}_{\textnormal{rotational}}(\mathscr{E}_i,\mathscr{E}_i^\prime)
  := \bigg\|\frac{\Hb_i\Hb_i^\top}{\Hb_i^\top\Hb_i}-\frac{\Hb_i^\prime\Hb_i^{\prime\top}}{\Hb_i^{\prime\top}\Hb_i^\prime}\bigg\|\,
\end{equation}
\begin{equation}\label{eq:t_d}
\mbox{d}_{\textnormal{translational}}(\Eq_i,\Eq_i^\prime) :=\bigg\|\frac{z_i\Hb_i}{\Hb_i^\top\Hb_i}-\frac{z_i^\prime\Hb_i^\prime}{\Hb_i^{\prime\top}\Hb_i^\prime}\bigg\|\,.
\end{equation}

%%Insight of two distances
For the affine subspace $\Eq_i$, we say that $\mathscr{L}_i:=\big\{\yb\in\R^m:{\Hb_i}^\top\yb=0\big\}$ is the subspace associated to  $\Eq_i$, and $\frac{z_i\Hb_i}{\Hb_i^\top\Hb_i}$ is the translational vector from  $\mathscr{L}_i$ to  $\Eq_i$.  Similarly, $\mathscr{L}_i^\prime$ can be identified for $\Eq_i^\prime$. The intuition and rational in introducing these two distances (\ref{eq:r_d}) and (\ref{eq:t_d}) are the following:
\begin{itemize}
  \item[(i)] The $\mbox{d}_{\textnormal{rotational}}$ is the gap of the two respective orthogonal projection operators onto $\mathscr{L}_i$ and $\mathscr{L}_i^\prime$, i.e.,
      \[\ba{rcl}
\mbox{d}_{\textnormal{rotational}} = \sup\limits_{\|\yb\|=1}\| \mathpzc{P}_{\mathscr{L}_i}(\yb)-\mathpzc{P}_{\mathscr{L}_i^\prime}(\yb)\|\,;
\ea\]
  \item[(ii)] The $\mbox{d}_{\textnormal{translational}}$ is the distance  between the two translational vectors of $\Eq_i,\Eq_i^\prime$.
\end{itemize}
A geometric illustration of both distances is shown in Figures \ref{fig-adj-a} and \ref{fig-adj-b}, respectively.

\begin{figure}[ht]
\centering
\begin{subfigure}{.4\textwidth}
\centering
		\includegraphics[width=1\linewidth]{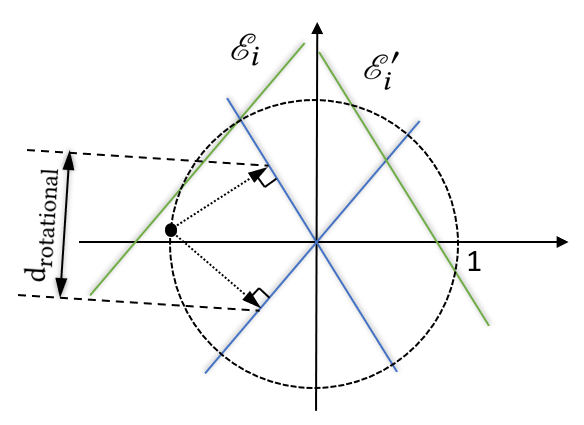}
\caption{The distance $\mbox{d}_{\textnormal{rotational}}$}
\label{fig-adj-a}
\end{subfigure}
\quad
\begin{subfigure}{.4\textwidth}
\centering
		\includegraphics[width=0.92\linewidth]{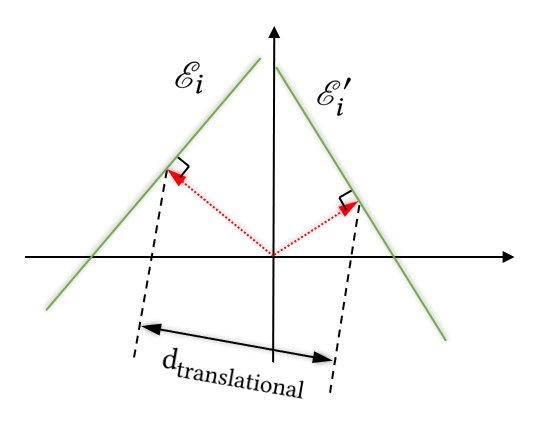}
\caption{The distance $\mbox{d}_{\textnormal{translational}}$}
\label{fig-adj-b}
\end{subfigure}
\caption{An illustration of the ``distance" between two affine subspaces $\Eq_i$ and $\Eq_i^\prime$. The affine subspaces are green lines, and their associated subspaces and translational vectors are denoted by blue lines and red dashed lines with arrows, respectively.
}
\label{fig-adj}
\end{figure}

In view of the previous analysis, we quantify the ``distance" between two affine subspaces $\Eq_i,\Eq_i^\prime$ by the sum of two quantities: $\mbox{d}_{\textnormal{rotational}}(\Eq_i,\Eq_i^\prime)$ defined in (\ref{eq:r_d}) and $\mbox{d}_{\textnormal{translational}}(\Eq_i,\Eq_i^\prime)$ defined in (\ref{eq:t_d}).
This further inspires the following definition.
\begin{definition}
For two linear equations $\mathfrak{E}:\Hb\yb=\zb,\ \mathfrak{E}^\prime:\Hb^\prime\yb=\zb^\prime$, we call them to be $\mu$-adjacent if
\[
\textnormal{d}_{\textnormal{rotational}}(\mathscr{E}_i,\mathscr{E}_i^\prime)+\textnormal{d}_{\textnormal{translational}}(\mathscr{E}_i,\mathscr{E}_i^\prime)\le\mu\,,\quad
\quad
i\in\mathrm{V}.
\]
\end{definition}

%%%%
\subsection{Distributed Computing Protocol}

In the following, the PPSC-Gossip network linear-equation (PPSC-Gossip-NLE) solver is presented.

	\begin{algorithm}[H]
\leftline{{\bf Input:} Local data $(\Hb_k,\zb_k)$, $k\in\mV$, initial value $\zeta_0\in\R^m$, and parameters $S,T,L\in\mathbb{N}_+$.}

\leftline{{\bf Initialize:} $s\gets0$, $l\gets0$ and $x_k(s)\gets \mathpzc{P}_{\Eq_k}(\zeta_0)$ for $k\in\mV$.}
	
\leftline{\mbox{$\quad$} {\bf For} $l = 0,1,\ldots, L-1$,  {\bf run}}

\leftline{\mbox{$\quad\qquad$} {\bf For} $s = l(S+T+1)+1,\ldots,l(S+T+1)+S$,  over $\mathrm{G}_{\rm p}$ {\bf run}}

\leftline{\mbox{$\quad\qquad \qquad$} $\xb_s=\mathsf{MultiGossipPPSC}\big(\xb_{s-1}\big)$}

\leftline{\mbox{$\quad\qquad$} {\bf For} $s= l(S+T+1)+S+1,\ldots,l(S+T+1)+S+T$, over $\mathrm{G}$ {\bf run}}

\leftline{\mbox{$\quad\qquad\qquad$} $x_k(s) = x_i(s-1)+a\sum\limits_{j\in {\rm N}_i} (x_j(s-1)-x_i(s-1))$ for $k\in\mV$.}

\leftline{\mbox{$\quad\qquad$} {\bf For} $s= l(S+T+1)+S+T+1$, {\bf run}}

\leftline{\mbox{$\quad\qquad\qquad$} $x_k(s) = \mathpzc{P}_{\Eq_k}(x_k(s-1))$  for $k\in\mV$.}

\leftline{{\bf Output:} $x_k(L(S+T+1))$, for $k\in\mV$.}
		\caption{PPSC-Gossip Network Linear-Equation (PPSC-Gossip-NLE) Solver}
	\end{algorithm}

%%Insight of Algorithm 4
The PPSC-Gossip-NLE solver is  comprised of $L$ recursions, each of which needs $S+T+1$ steps and consists of three different procedures in order: (i) PPSC-Gossip procedure with $S$ steps; (ii) Average-Consensus procedure with $T$ steps; (iii) one-step orthogonal projection. The procedures (i) and (ii) together are exactly the PPSC-Gossip-AC algorithm. At the procedure (iii), the computation is self performed by each agent to orthogonally project its state to its own affine subspace, where the local private datasets are encoded.

\subsection{Computation Accuracy and Differential Privacy}

Let  $\chi_l=[\xb_{l(S+T+1)+S};\xb_{l(S+T+1)+S+1};\ldots;\xb_{l(S+T+1)+S+T}]$ and $\chi_{\rm ea} = [\chi_0;\ldots;\chi_{L-1}]$. Define the sequence of all selected communication edges at the PPSC-Gossip procedures as $\mathcal{E}_{\rm p}=\{\mathcal{E}_{{\rm p},0},\ldots,\mathcal{E}_{{\rm p},L-1}\}$ with $\mathcal{E}_{{\rm p},l}=\big(\mathsf{E}_{{\rm p},l}^1,\ldots,\mathsf{E}_{{\rm p},l}^q\big)$ and
$
\mathsf{E}_{{\rm p},l}^{k}=\left(\mathsf{e}_{{\rm p}}^k(l(S+T+1)+1),\ldots,\mathsf{e}_{{\rm p}}^k(l(S+T+1)+S)\right).
$
Define the map from the network linear equations $\mathfrak{E}$ to $\chi_{\rm ea}$  as $\chi_{\rm ea} = \mathscr{M}_{\rm eq}\big(\mathfrak{E}\big)$.
As in the previous sections, let an eavesdropper with  full observations of the selected edges (i.e., $\mathcal{E}_{\rm p}$) at the PPSC-Gossip procedure and the node-to-node communications (i.e.,$\chi_{\rm ea}$) over the public network $\mathrm{G}$. In the following, we specify the differential privacy of the PPSC-Gossip-NLE solver.

\begin{definition}
   Let $\mu>0$, $\epsilon>0$ and $0<\delta<1/2$. The $\mathscr{M}_{\rm eq}$ is termed to be \emph{$(\epsilon,\delta)$-differentially private under $\mu$-adjacency} if for all $R\subseteq \range(\mathscr{M}_{\rm eq})$, there holds
  \beeq{\ba{l}
  \mathbb{P}\left(\mathscr{M}_{\rm eq}\big(\mathfrak{E} \,|\, \mathcal{E}_{\rm p}\,\big)\in R\right) \leq e^{\epsilon}\mathbb{P}\left(\mathscr{M}_{\rm eq}\big(\mathfrak{E}^\prime \,|\, \mathcal{E}_{\rm p}\big)\in R\right) + \delta
  \ea}
  for any two $\mu$-adjacent equations $\mathfrak{E}^\prime, \mathfrak{E}$.
\end{definition}

%%Theorem
Let $\lambda_H=\sigma_M\big(\Ib_m -\frac{1}{n}\sum_{i=1}^n \frac{\Hb_i\Hb_i^\top}{\Hb_i^\top\Hb_i}\big)$.
As $\rank(\Hb)=m$ by assumption, it can be seen that matrix $\frac{1}{n}\sum_{i=1}^n \frac{\Hb_i\Hb_i^\top}{\Hb_i^\top\Hb_i}$ is strictly positive definite with all eigenvalues within the unit circle. Thus, $1>\lambda_H\geq 0$.
Let
\[\ba{l}
L^\ast(\nu) = {\left(\log\nu - \log(2n\|\zeta_0-\yb^\ast\|^2)\right)}/{\big(\log(\varepsilon_0+\lambda_H^2)-\log(1+ \varepsilon_0)\big)}\\
S^\ast(\rho,L) = \Big({\log(1-\rho^{\frac{1}{2qL}})-\log n_{\max}}\Big)/\left({\log(1-\frac{1+r^\dag}{n_{\max}})}\right)\,\\
\sigma_\gammab^\ast(\epsilon,\nu,L,S) = {\mu\kappa(\frac{\epsilon}{L},\delta_{\sharp})(\sqrt{\nu}+\phi_{\nu}+\sqrt{n})}/ {\lambda_{\rm ppsc}}
\ea\]
with $1>\varepsilon_0>0$, $\phi^\ast(\nu)=2\sqrt{n}\|\yb^\ast\|+ \sqrt{n}\|\zeta_0\|  + \frac{2-\lambda_H}{1-\lambda_H}\sqrt{\nu}$ and $\delta_{\sharp}(\epsilon,\delta,L)
=(\delta+e^{\epsilon})^{1\over L}-\exp(\frac{\epsilon}{L})$, and
\[\ba{rcl}
  T^\ast(\rho,\nu,L,S,\sigma_{\gammab}) = \frac{1}{\log(1-\lambda_{\rm G})}\min\left\{\frac{1}{2}\log\frac{(1-\rho^{1\over 2L})\nu}{{\phi^\ast}^2 + 2q^2S^2\sigma_{\gammab}^2}, \log\frac{1-\lambda_H^2}{5n(1+{1\over \varepsilon_0})},  \log\frac{\nu(1-\lambda_H^2)}{16(n\|\yb^\ast\|^2+S^2q^2\sigma_{\gammab}^2)}\right\}\,.
\ea\]
Next, we show that the PPSC-Gossip-NLE solver can solve the linear equations with any accuracy level, while achieving an arbitrary $(\epsilon,\delta)$-differential privacy with any prescribed probability.

\begin{theorem}\label{theorem-Equ}
For any $\mu>0$, $\epsilon>0$, $0<\delta<1/2$, $\rho>0$ and $\nu>0$, let $L \geq L^\ast$, $S\geq S^\ast$, $\sigma_\gammab\geq \sigma_\gammab^\ast$,  and $T\geq T^\ast$. Then, the PPSC-Gossip-NLE solver
  \begin{itemize}
    \item[(i)] preserves $(\epsilon,\delta)$-differential privacy under $\mu$-adjacency with a probability that is at least $\rho$, and
    \item[(ii)] computes the solution with a $\nu$-accuracy, i.e., $\mathbb{E}\left(\|\xb_{L(S+T+1)}- ({\bf 1}_n\otimes \yb^\ast)\|^2\right) \leq \nu$.
  \end{itemize}

\end{theorem}

\section{PPSC-Gossip Distributed Convex Optimization}
\label{sec-5}

In this section, the PPSC-Gossip-AC algorithm is explored to solve the distributed convex optimization problem (\ref{eq:opti}), i.e., each agent $i$ holds a private convex function $f_i(\cdot):\R^m \rightarrow \R$, and all agents aim to solve the  optimization problem
$\min_{\yb\in \mathtt{C}} f(\yb):=\sum_{i=1}^n f_i(\yb)$, where
$\mathtt{C}\subset\R^m$ is a   compact convex set.

In order to facilitate a convenient discussion, we assume all functions $f_i$, $i=1,\ldots,n$ are parameterized in the form of $f_i(\yb)=g_i(\tau_i,\yb)$, where parameter vectors $\tau_i\in\mathcal{T}$ are privacy-sensitive with a public bounded set $\mathcal{T}\subset\R^{d}$, and  functions $g_i:\R^{d}\times \R^m\rightarrow\R$ are continuously differentiable. Without loss of generality we assume that the  structure (form) of $g_i$ is public and the privacy sensitivity takes place at the parameters $\tau_i$.
The adjacency of two functions is  characterized as below.
\begin{definition}
Let $f_i(\yb)=g_i(\tau_i,\yb)$ and $f_i^\prime(\yb)=g_i(\tau_i^\prime,\yb)$ with $\tau_i, \tau_i^\prime\in\mathcal{T}$ for $i\in\mathrm{V}$. We say $\mathbf{F}=[f_1(\yb);\ldots;f_n(\yb)]$ and $\mathbf{F}^\prime =[f_1^\prime(\yb);\ldots;f_n^\prime(\yb)]$ to be $\mu$-adjacent if $\|\tau_i-\tau_i^\prime\| \leq \mu$ holds for $i\in\mathrm{V}$.
\end{definition}

\subsection{Distributed Computing Protocol}

Let  $\mathpzc{P}_{\mathtt{C}}(\yb)$ be the projection of $\yb\in\R^m$ on the set $\mathtt{C}$, i.e.,
$
\mathpzc{P}_{\mathtt{C}}(\yb) = \arg \min_{\yb^\prime\in \mathtt{C}}\|\yb-\yb^\prime\|\,.
$
In the following, the PPSC-Gossip distributed convex optimization (PPSC-Gossip-DCO) algorithm is presented.
\begin{algorithm}[H]
\leftline{{\bf Input:} The subgradient $\nabla f_k(\cdot)$,  stepsize $\alpha_l>0$,  initial value $\zeta_0\in \mathtt{C}$ and parameters $S,T,L\in\mathbb{N}_+$.}

\leftline{{\bf Initialize:} $s\gets0$, $l\gets0$ and $x_k(s)\gets \mathpzc{P}_{\mathtt{C}}\big(\zeta_0-\alpha_0\nabla f_k(\zeta_0)\big)$.}
	
\leftline{\mbox{$\quad$} {\bf For} $l = 0,1,\ldots, L-1$,   over $\mathrm{G}_{\rm p}$ {\bf run}}

\leftline{\mbox{$\quad\qquad$} {\bf For} $s = l(S+T+1)+1,\ldots,l(S+T+1)+S$,  {\bf run}}

\leftline{\mbox{$\quad\qquad \qquad$} $\xb_s=\mathsf{MultiGossipPPSC}\big(\xb_{s-1}\big)$}

\leftline{\mbox{$\quad\qquad$} {\bf For} $s = l(S+T+1)+S+1,\ldots,l(S+T+1)+S+T$, over $\mathrm{G}$ {\bf run}}

\leftline{\mbox{$\quad\qquad\qquad$} $x_k(s) = x_i(s-1)+a\sum\limits_{j\in {\rm N}_i} (x_j(s-1)-x_i(s-1))$ for $k\in\mV$.}

\leftline{\mbox{$\quad\qquad$} {\bf For} $s = l(S+T+1)+S+T+1$, at each node $k\in\mV$ {\bf run}}

\leftline{\mbox{$\quad\qquad\qquad$} $x_k(s) = \mathpzc{P}_{\mathtt{C}}\big(x_k(s-1)-\alpha_{l+1}\nabla f_k(x_k(s-1))\big)$   for $k\in\mV$.}

\leftline{{\bf Output:} $x_k(L(S+T+1))$, for $k\in\mV$.}
		\caption{PPSC-Gossip Distributed Convex Optimization (PPSC-Gossip-DCO) Algorithm}
	\end{algorithm}

As in the PPSC-Gossip-NLE algorithm, the above PPSC-Gossip-DCO algorithm is also comprised of $L$ recursions, each of which needs $S+T+1$ steps and consists of three different procedures in order: (i) PPSC-Gossip procedures with $S$ steps; (ii) Average-Consensus procedures with $T$ steps; (iii) one-step projected subgradient descent. The procedures of (i) and (ii) together are exactly the PPSC-Gossip-AC algorithm, while the third procedure is self performed by each agent to take the subgradient descent and then projection on the convex set $\mathtt{C}$. In this way, the local private functions $f_k$ are encoded in the update rules in the form of the subgradients.

\subsection{Computation Accuracy and Differential Privacy}

We follow the definitions of $\mathcal{E}_{\rm p}$ and $\chi_{\rm ea}$ in Section \ref{sec-eq}, and define the map from the functions $\mathbf{F} :=[f_1;f_2;\ldots;f_n]$ to $\chi_{\rm ea}$  as $\chi_{\rm ea} = \mathscr{M}_{\rm op}\big(\mathbf{F} \big)$. As in the previous sections, let an eavesdropper with  full observations of the selected edges (i.e., $\mathcal{E}_{\rm p}$) at the PPSC-Gossip procedure and the node-to-node communications (i.e.,$\chi_{\rm ea}$) over the public network $\mathrm{G}$. We then specify the notion of differential privacy for the PPSC-Gossip-DCO algorithm.
\begin{definition}
Let $\mu>0$, $\epsilon>0$ and $0<\delta<1/2$. The $\mathscr{M}_{\rm op}$ is termed to be $(\epsilon,\delta)$-differentially private under $\mu$-adjacency if for all $R\subseteq \range(\mathscr{M}_{\rm op})$, there holds
  \beeq{\ba{l}
  \mathbb{P}\left(\mathscr{M}_{\rm op}\big(\mathbf{F} \,|\, \mathcal{E}_{\rm p}\big)\in R\right)  \leq e^{\epsilon}\mathbb{P}\left(\mathscr{M}_{\rm op}\big(\mathbf{F}^\prime \,|\, \mathcal{E}_{\rm p}\big)\in R\right) + \delta
  \ea}
  for any two  $\mu$-adjacent $\mathbf{F}^\prime, \mathbf{F}$.
\end{definition}

We denote the optimal solution set of (\ref{eq:opti}) as $\mathtt{C}^\dag\subseteq \mathtt{C}$.
Let $\phi^\dag=\max\limits_{\yb\in \mathtt{C}}\|\yb\|$ and $$
g^\dag =  \max\limits_{(i,\tau_i,\yb)\in\mV\times\mathcal{T}\times\Omega_{\nu}} \left\|\frac{\partial g_i}{\partial \yb}(\tau_i,\yb)\right\|
$$ with $\Omega_\nu:=\{\yb\in\R^m: \|\yb\|^2\leq \max\limits_{\yb^\prime\in \mathtt{C}}\|\yb^\prime\|^2 + \nu\}$.
Let
\[\ba{l}
S^\dag(\rho,L) = \Big({\log(1-\rho^{\frac{1}{qL}})-\log n_{\max}}\Big)/{\log(1-\frac{1+r^\dag}{n_{\max}})}\,\\
\sigma_\gammab^\dag(\mu,\epsilon,\delta,L) = {n\mu g^\dag\mathcal{R}(\frac{\epsilon}{L},\delta_\sharp)} /{\lambda_{\rm ppsc}}\,\\
T^\dag(p,\nu,L,S,\sigma_{\gammab})= \Big({\log((1-p^{1\over L})\nu\alpha_L^4)-\log(n{\phi^\dag}^2 + 2q^2S^2\sigma_{\gammab}^2)}\Big)/{\log(1-\lambda_{\rm G})^2}\,.
\ea\]
Regarding the differential privacy and computation accuracy of the PPSC-Gossip-DCO algorithm, the following result is formulated.
\begin{theorem}\label{theorem-Op}
Let the stepsize $\alpha_l=\frac{1}{l+1}$. For any $\mu>0$, $\epsilon>0$, $0<\delta<1/2$, $0<\rho<1$, $0<p<1$ and $\nu>0$, there exists an $L^\dag>0$ such that for all $L\geq L^\dag$, $S\geq S^\dag$, $\sigma_\gammab\geq \sigma_\gammab^\dag$ and $T\geq T^\dag$, the PPSC-Gossip-DCO algorithm
  \begin{itemize}
    \item[(i)] computes a $\nu$-accuracy  with a probability that is at least $p$, i.e.,
  \beeq{
  \mathbb{P}\left(\|\xb_{L(S+T+1)}-\1_n\otimes\yb^\dag\|^2\leq \nu\right)\geq  p\,,\quad \mbox{for some $\yb^\dag\in \mathtt{C}^\dag$}
  }
    \item[(ii)] preserves $(\epsilon,\delta)$-differential privacy under $\mu$-adjacency with a probability that is at least $\rho$.
  \end{itemize}
\end{theorem}

  Regarding the design of the stepsize $\alpha_l$, we note that any $\alpha_l$ satisfying (i) $\lim\limits_{l\rightarrow\infty}\alpha_l=0$; (ii) $\lim\limits_{l\rightarrow\infty}\sum\limits_{i=0}^l\alpha_i^2 < \infty$; (iii) $\lim\limits_{l\rightarrow\infty}\sum\limits_{i=0}^l\alpha_i=\infty$ is feasible \cite{nedic2010constrained}.
  We also remark that the distributed computing protocol in Algorithm 5 can be adapted to compute the unique least-squares solution of the network linear equations (\ref{eq:LINEAR ALGEBRAIC EQUATION}). The idea, motivated by \cite{shi2017networkTAC}, lies in modifying the procedure of projected subgradient descent at $s=(l+1)(S+T+1)-1$ to the following
  \[
  x_k(s) = x_k(s-1)+\alpha_{l+1}\big(\mathpzc{P}_{\Eq_k}(x_k(s-1))-x_k(s-1)\big)\,,\quad \mbox{ $\forall\,\,k\in{\rm V}$}\,.
  \]
  Under the resulting distributed computing protocol, all results in Theorem 4 can still be preserved, i.e., the unique least-squares solution can be computed with an arbitrarily given accuracy, while achieving an arbitrary $(\epsilon,\delta)$-differential privacy with any prescribed probability.

\section{Case Studies}
In this section, we provide a series of numerical examples that illustrate the effectiveness of our results.
\subsection{Averaging Consensus}
\label{subsec-av}

In this subsection, numerical simulations are conducted to demonstrate the feasibility of the proposed PPSC-Gossip-AC algorithm.
We consider a system of 10 agents over an public cycle network ${\rm G}$  as in Figure \ref{fig-ring-a}, where each edge is assigned with weight $1\over 10$ and each agent $i$ holds a sensitive number $d_i$ (see Table \ref {tab:data-av}).
%In this setting, we design the private graph ${\rm G}_{\rm p}$, run the PPSC-Gossip averaging consensus algorithm over networks ${\rm G}_{\rm p},{\rm G}$, and report the relation between the computation accuracy $\nu$, the privacy budgets $\epsilon,\delta,\rho$, and the running steps $S$ of the multi-gossiping PPSC mechanism and $T$ of the averaging consensus algorithm.

\begin{table}
  \caption{Private datasets for average consensus}
  \centering
  \begin{tabular}{lllll}
    \hline
    $d_1=10$ & $d_2=100$ & $d_3=20$ & $d_4=-30$ & $d_5=-20$ \\
    $d_6=60$ & $d_7=70$ & $d_8=0$ & $d_9=80$ & $d_{10}=-20$\\
    \hline
\end{tabular}
\label{tab:data-av}
\end{table}

\begin{figure}[ht]
\centering
\begin{subfigure}{.2\textwidth}
\centering
		\includegraphics[width=.8\linewidth]{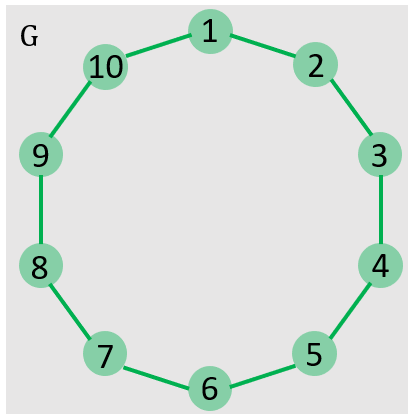}
\caption{ }
\label{fig-ring-a}
\end{subfigure}
%\quad
%\begin{subfigure}{.2\textwidth}
%\centering
%		\includegraphics[width=.8\linewidth]{fig_cycle_10_private.png}
%\caption{ }
%\label{fig-ring-b}
%\end{subfigure}
\quad
\begin{subfigure}{.2\textwidth}
\centering
		\includegraphics[width=0.8\linewidth]{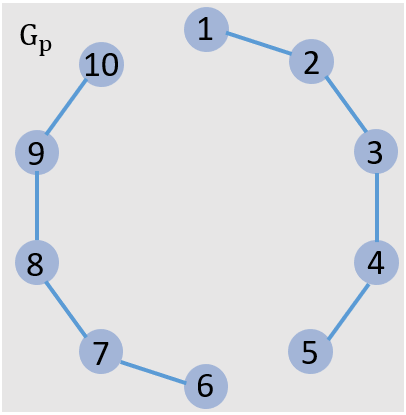}
\caption{}
\label{fig-G-p-0}
\end{subfigure}
\quad
\begin{subfigure}{.2\textwidth}
\centering
		\includegraphics[width=0.8\linewidth]{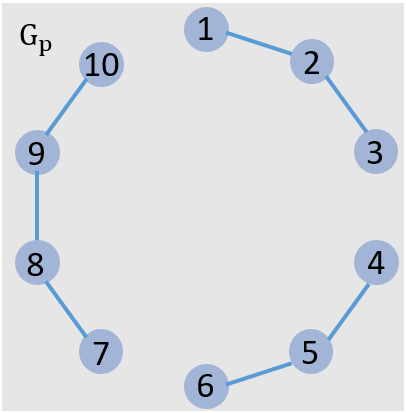}
\caption{}
\label{fig-G-p-1}
\end{subfigure}
\caption{The public  network ${\rm G}$ and the private network ${\rm G}_{\rm p}$ with $n=10$.
}
\label{fig-ring}
\end{figure}

In order to assess the effect of ${\rm G}_{\rm p}$, we run the multi-gossiping PPSC mechanism over three kinds of ${\rm G}_{\rm p}$: Figures \ref{fig-ring-a},  \ref{fig-G-p-0} and \ref{fig-G-p-1}, having one, two and three components, respectively,  and show the relationship between the iteration step $S$ and the differential privacy probability $\rho$, the latter of which is indeed the probability of the event $\mathcal{Q}_S$ that all node states have altered after the $S$-step multi-gossiping PPSC mechanism.
Figure \ref{fig-rho-av} shows that a larger $\rho$ requires a larger $S$ for all three graphs ${\rm G}_{\rm p}$, and given any $\rho$, the smallest $S$ is needed by the ${\rm G}_{\rm p}$ in Figure \ref{fig-G-p-1}.
Thus, in the following simulations ${\rm G}_{\rm p}$ is fixed as in Figure \ref{fig-G-p-1}.
Under this ${\rm G}_{\rm p}$, we then compare the minimal $S$ required practically (i.e., in Figure \ref{fig-rho-av}) and theoretically (i.e., in Theorem \ref{theorem-Ave}) to fulfill the requirement of probability $\rho$.  Table \ref{tab:rho-av} indicates that the theoretic $S$ in Theorem \ref{theorem-Ave} is indeed more conservative under each $\rho$. In the following simulations, we let the desired differential privacy probability $\rho= 0.9998$, which can be guaranteed by $S=25$ as in Table \ref{tab:rho-av}.

Throughout all simulations,  following the standard differential privacy guideline \cite{dwork2014algorithmic}, we fix $\mu=1$ and  let $\delta=10^{-6}$. With these $\mu,\delta$, we consider various privacy levels $\epsilon=10^{-3},10^{-2},10^{-1}$ and run the PPSC-Gossip-AC algorithm  with the noise variance $\sigma_\gammab$ chosen as the minimal value $\sigma_\gammab^\ast$ in Theorem \ref{theorem-Ave}.  Figure \ref{fig-case-av} shows the performance of the algorithm under different privacy requirements. The computation accuracy  strictly decreases as $T$ increases for  all privacy levels $\epsilon$, while a larger averaging step $T$ is required under a smaller privacy level $\epsilon$  to reach the same computation accuracy.  To show whether the theoretic $T$ in Theorem \ref{theorem-Ave} can fulfill the desired computation accuracy, we compare it with the practical one in Figure \ref{fig-case-av} to reach various accuracy $\nu$. Table \ref{tab:T-ave} indicates that under each pair of $(\epsilon,\nu)$, the $T$ in Theorem \ref{theorem-Ave} is  larger and thus guarantees the desired computation accuracy.
We also compare our PPSC-Gossip-AC algorithm with the He2020 algorithm in \cite{He2020TSP}, whose computation accuracy under different privacy requirements $\epsilon$ is presented in Figure \ref{fig-case-av-He}. As seen from Figure \ref{fig-case-av-He}, a higher privacy requirement results in a larger computation error. This indicates a trade-off between the computation accuracy and privacy in \cite{He2020TSP}, while such a trade-off does not exist in PPSC-Gossip-AC algorithm.

%\begin{figure}
%\centering
%\begin{subfigure}{.35\textwidth}
%\centering
%		\includegraphics[width=0.6\linewidth]{G_p-av-1.png}
%\caption{}
%\label{fig-G-p-0}
%\end{subfigure}
%\quad
%\begin{subfigure}{.35\textwidth}
%\centering
%		\includegraphics[width=0.58\linewidth]{G_p-av-2.png}
%\caption{}
%\label{fig-G-p-1}
%\end{subfigure}
%
%\caption{Two structures of ${\rm G}_{\rm p}$. }
%\label{fig-G-p}
%\end{figure}
%

\begin{figure}
  \centering
  \includegraphics[width=7cm]{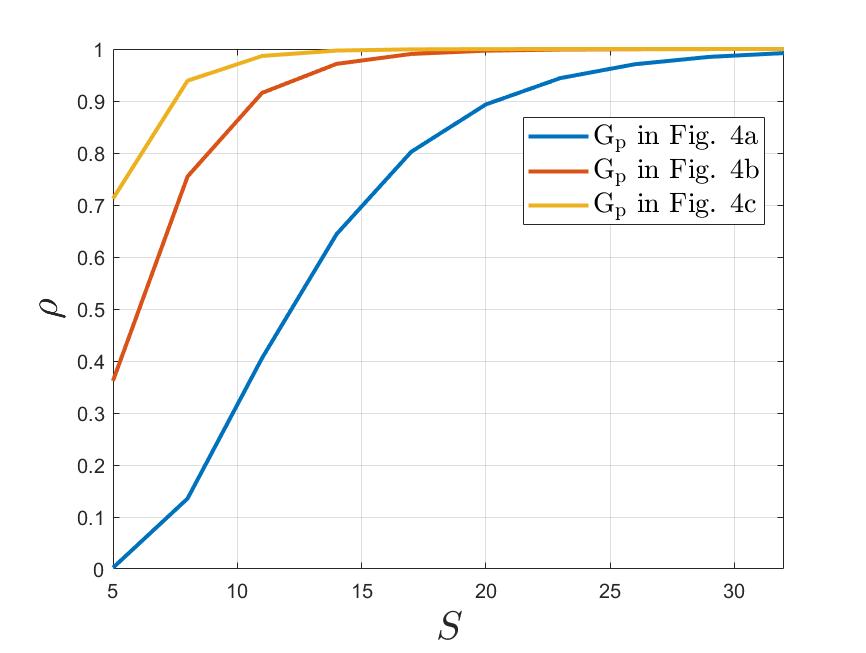}
  \caption{Relation between $S$ and $\rho$ under different ${\rm G}_{\rm p}$ ($10^6$ samples).}
%  \Description{Private and public Graphs.}
\label{fig-rho-av}
\end{figure}

\begin{table}
    \centering
    \caption{A comparison of minimal $T$ to reach various accuracy $\nu$ in Theorem \ref{theorem-Ave} and Figure \ref{fig-case-av}}
    \begin{subtable}[t]{2in}
  \begin{tabular}{lcc}
    \hline
     $\rho$ &Theo. \ref{theorem-Ave} & Fig. \ref{fig-rho-av}\\
    \hline
    $0.7126$ & 8 & 5\\
    $0.9393$ & 12 & 8\\
    $0.9867$ & 15 & 11\\
  \hline
\end{tabular}
    \end{subtable}
  \qquad
    \begin{subtable}[t]{2in}
  \begin{tabular}{lcc}
    \hline
     $\rho$ &Theo. \ref{theorem-Ave} & Fig. \ref{fig-rho-av}\\
    \hline
    $0.9970$ & 18 & 14\\
    $0.9993$ & 21 & 17\\
    $0.9998$ & 25 & 20\\
  \hline
\end{tabular}

    \end{subtable}

\label{tab:rho-av}
\end{table}

%\begin{table}
%  \caption{Minimal $S$ between Theorem \ref{theorem-Ave} and Figure \ref{fig-rho-av}}
%  \centering
%  \begin{tabular}{lcc}
%    \hline
%     $\rho$ &Theorem \ref{theorem-Ave} & Figure \ref{fig-rho-av}\\
%    \hline
%    $0.7126$ & 8 & 5\\
%    $0.9393$ & 12 & 8\\
%    $0.9867$ & 15 & 11\\
%    $0.9970$ & 18 & 14\\
%    $0.9993$ & 21 & 17\\
%    $0.9998$ & 25 & 20\\
%  \hline
%\end{tabular}
%\label{tab:rho-av}
%\end{table}

%This in fact is consistent with Theorem 1 that for arbitrary privacy budgets and computation accuracy, there exists $T$ such that both the $(\epsilon,\delta)$-differential privacy and the desired computation accuracy can be achieved simultaneously.
%
%
%
%\begin{figure}
%  \centering
%  \includegraphics[width=7cm]{fig-error-ave-2.jpg}
%  \caption{The computation accuracy $\mathbb{E}\|\xb_{S+T}-1_n\otimes d^\ast\|^2$.}
%%  \Description{Private and public Graphs.}
%\label{fig-case-av}
%\end{figure}
%
%\begin{figure}
%  \centering
%  \includegraphics[width=7cm]{fig-error-ave-acc.jpg}
%  \caption{Evolutions of $\mathbb{E}\|\xb_{s}-1_n\otimes d^\ast\|^2$ under $\epsilon=10^{-3}$.}
%%  \Description{Private and public Graphs.}
%\label{fig-case-av-acc}
%\end{figure}

\begin{figure}[ht]
\begin{subfigure}{.4\textwidth}
\centering
	\includegraphics[width=1.0\linewidth]{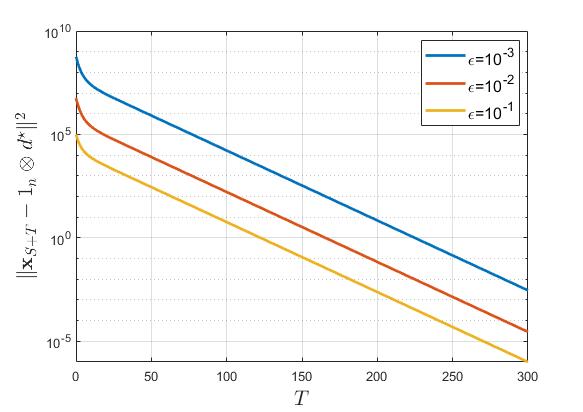}
\caption{The PPSC-Gossip-AC algorithm}
\label{fig-case-av}
\end{subfigure}
\qquad
\begin{subfigure}{.4\textwidth}
\centering
		\includegraphics[width=1.0\linewidth]{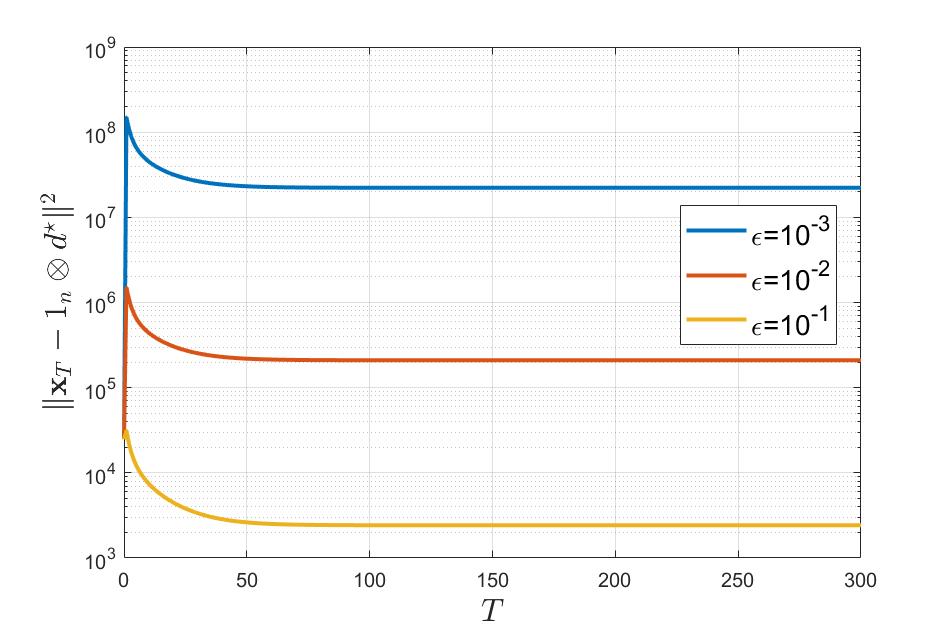}
\caption{The algorithm of He2020 \cite{He2020TSP}}
\label{fig-case-av-He}
\end{subfigure}
\caption{The computation accuracy in the mean square sense}
\end{figure}

\begin{table}
    \centering
    \caption{A comparison of minimal $T$ to reach various accuracy $\nu$ in Theorem \ref{theorem-Ave} and Figure \ref{fig-case-av}}
    \begin{subtable}[t]{2in}
        \begin{tabular}{ccc}
        \hline
        $\nu$ & Theo. \ref{theorem-Ave}& Fig. \ref{fig-case-av}\\
        \hline
        $10^0$ & 341 & 226 \\
        $10^{-1}$ & 371 & 255\\
        $10^{-2}$ & 400 & 285 \\
        \hline
        \end{tabular}
        \caption{$\epsilon=10^{-3}$}
    \end{subtable}
  \qquad
    \begin{subtable}[t]{2in}
        \begin{tabular}{ccc}
        \hline
        $\nu$ & Theo. \ref{theorem-Ave}& Fig. \ref{fig-case-av}\\
        \hline
        $10^0$ & 282 & 167  \\
        $10^{-1}$ & 312 & 197\\
        $10^{-2}$ & 341 & 227 \\
        \hline
        \end{tabular}
        \caption{$\epsilon=10^{-2}$}
    \end{subtable}
    \qquad
    \begin{subtable}[t]{2in}
        \begin{tabular}{ccc}
        \hline
        $\nu$ & Theo. \ref{theorem-Ave}& Fig. \ref{fig-case-av}\\
        \hline
        $10^0$ & 223 & 123\\
        $10^{-1}$ & 253 &  153\\
        $10^{-2}$ & 282 &  183\\
        \hline
        \end{tabular}
        \caption{$\epsilon=10^{-1}$}
    \end{subtable}
   \label{tab:T-ave}
\end{table}

%%\begin{table}
%%    \centering
%%    \caption{Minimal $T$ to reach various accuracies $\nu$ under $\epsilon=0.001$ in Theorem \ref{theorem-Ave} and Figure \ref{fig-case-av}}
%%  \begin{tabular}{ccc}
%%    \hline
%%     $\nu$ & Theorem \ref{theorem-Ave}& Figure \ref{fig-case-av}\\
%%    \hline
%%   $10^0$ & 133 & 103 \\
%%   $10^{-1}$ & 144 & 114\\
%%   $10^{-2}$ & 156 & 126 \\
%%  \hline
%%\end{tabular}
%%\end{table}
%%
%%\begin{table}
%%
%%  \centering
%%  \caption{Minimal $T$ to reach various accuracies under $\epsilon=0.01$ in Theorem \ref{theorem-Ave} and Figure \ref{fig-case-av}}
%%  \begin{tabular}{ccc}
%%    \hline
%%     $\nu$ & Theorem \ref{theorem-Ave}& Figure \ref{fig-case-av}\\
%%    \hline
%%   $10^0$ & 110 & 80  \\
%%   $10^{-1}$ & 121 & 91\\
%%   $10^{-2}$ & 133 & 103 \\
%%  \hline
%%\end{tabular}
%%\end{table}
%%
%%\begin{table}
%%
%%  \centering
%%  \caption{Minimal $T$ to reach various accuracies under $\epsilon=0.1$ in Theorem \ref{theorem-Ave} and Figure \ref{fig-case-av}}
%%  \begin{tabular}{ccc}
%%    \hline
%%     $\nu$ & Theorem \ref{theorem-Ave}& Figure \ref{fig-case-av}\\
%%    \hline
%%   $10^0$ & 87 & 63\\
%%   $10^{-1}$ & 98 &  74\\
%%   $10^{-2}$ & 110 &  86\\
%%  \hline
%%\end{tabular}
%%
%%\end{table}
%
%

\subsection{Network Linear Equations}

In this subsection, we apply the PPSC-Gossip-NLE solver to solve the network linear algebraic equation over the public network ${\rm G}$ in Figure \ref{fig-ring}, where each edge is assigned with weight $1\over 4$ and each agent $i$ holds a sensitive equation $\mathfrak{E}_i$ (see Table \ref{tab:data-LAE}).

%We let $\|\yb^\ast\|\leq 20$ and $\varepsilon_0=0.01$.
We consider three computation accuracy requirements $\nu=10^0, 10^{-1}, 10^{-2}$, and thus select $L=207$, $L=260$ and $L=313$, respectively according to Theorem \ref{theorem-Equ}.
For these cases, we fix the differential privacy probability as $\rho=0.95$ and thus let $S=26$ for all simulations.
For each accuracy requirement $\nu=10^0, 10^{-1}, 10^{-2}$ and privacy requirement $\epsilon=10^{-3},10^{-2},10^{-1}$, we run the PPSC-Gossip-NLE solver with the private network ${\rm G}_{\rm p}$ in Figure \ref{fig-G-p-1} and the noise variance $\sigma_\gammab$ as the minimal value $\sigma_\gammab^\ast$ in Theorem \ref{theorem-Equ}. The simulation results are presented in Figure \ref{fig-case-eq} and Tables \ref{tab:T-eq}-\ref{tab:prob-eq}. Figures \ref{fig-case-eq-0}, \ref{fig-case-eq-1} and \ref{fig-case-eq-2} demonstrate the relationship between the resulting computation accuracy and the averaging step $T$ at various pairs of accuracy and privacy  requirements, i.e., $(\nu,\epsilon)$.
In all these figures, it can be seen that as $T$ increases, the resulting computation accuracy decreases until a lower bound that is smaller than the expected accuracy $\nu$.
Under each pair of $(\nu,\epsilon)$, from Table \ref{tab:T-eq} the theoretic $T$ in Theorem \ref{theorem-Equ} is larger than that in Figures \ref{fig-case-eq-0}, \ref{fig-case-eq-1} and \ref{fig-case-eq-2}, and from Table \ref{tab:prob-eq} the resulting differential privacy probability is larger than the expected $\rho=0.95$. Therefore, the proposed solver following Theorem \ref{theorem-Equ} achieves  the desired computation accuracy and  differential privacy, simultaneously.
%Taking the probability of the differential privacy into account, as shown in Table 7, the resulting probabilities using the MCM are smaller than the corresponding expected probabilities $\rho$, respectively, i.e., the desired differential privacy is achieved following the design in Theorem \ref{theorem-Equ}.

\begin{table}
  \centering
  \caption{Private datasets for solving network linear equations with solution $\yb^\ast=[5;-10;10;-5;1;5]$}
  \begin{tabular}{ll}
    \hline
    $\Hb_1 = [1;2;0;0;0;0]$; $z_1 = -15$ & $\Hb_6 = [2;0;1;0;2;1]$; $z_6 = 27$\\
    $\Hb_2 = [1;1;1;0;0;0]$; $z_2 = 5$ & $\Hb_7 = [1;1;1;2;0;1]$; $z_7 = 0$\\
    $\Hb_3 = [0;1;1;0;0;3]$; $z_3 = 15$ & $\Hb_8 = [3;1;5;6;8;-2]$; $z_8 = 23$ \\
    $\Hb_4 = [0;-1;1;2;5;-2]$; $z_4 = 5$ & $\Hb_9 = [0;-2;0;1;5;0]$; $z_9 = 20$\\
    $\Hb_5 = [5;-2;0;2;0;1]$; $z_5 = 40$ & $\Hb_{10} = [0;0;0;0;2;-1]$; $z_{10} = -3$ \\
  \hline
\end{tabular}
  \label{tab:data-LAE}
\end{table}

\begin{figure}[ht]
\centering
 \qquad \qquad \qquad
\begin{subfigure}{.4\textwidth}
\centering
		\includegraphics[width=1.15\linewidth]{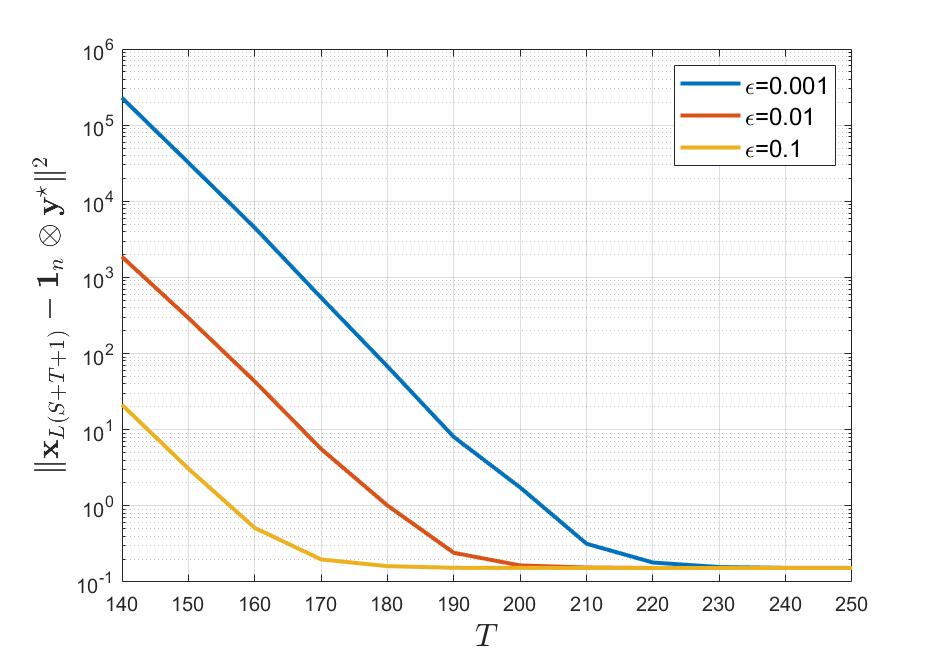}
\caption{Expected accuracy $\nu=10^0$}
\label{fig-case-eq-0}
\end{subfigure}
\newline
\begin{subfigure}{.4\textwidth}
\centering
		\includegraphics[width=1.15\linewidth]{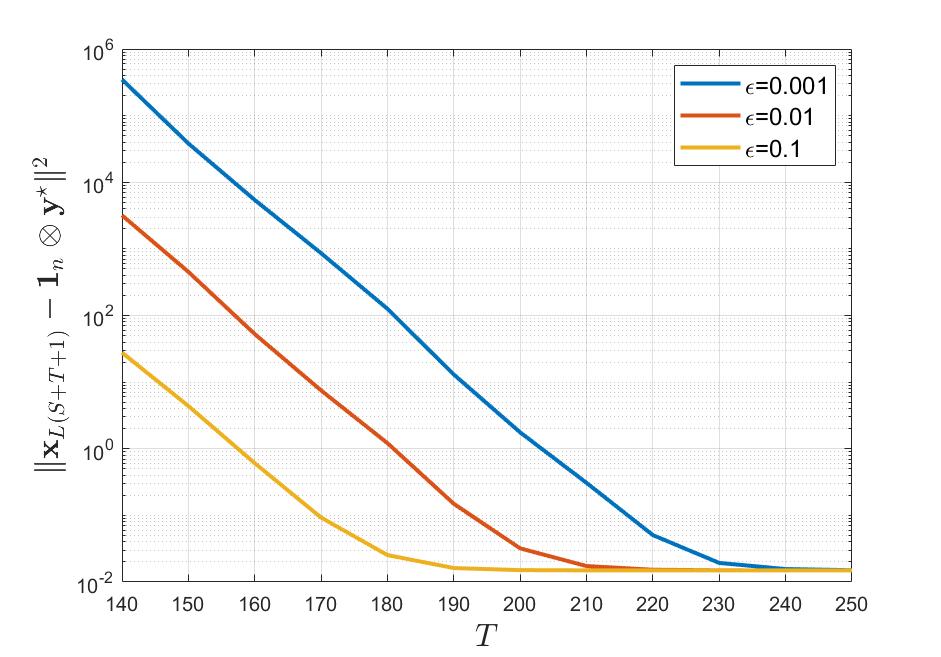}
\caption{Expected accuracy $\nu=10^{-1}$}
\label{fig-case-eq-1}
\end{subfigure}
\qquad \qquad
\begin{subfigure}{.4\textwidth}
\centering
		\includegraphics[width=1.15\linewidth]{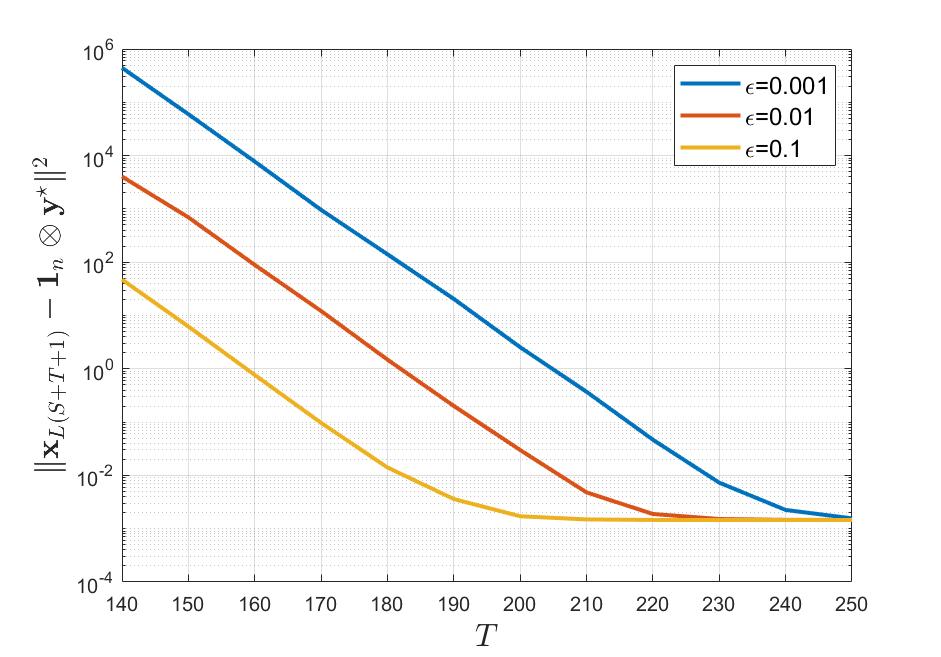}
\caption{Expected accuracy $\nu=10^{-2}$}
\label{fig-case-eq-2}
\end{subfigure}
\caption{The actual accuracy $\mathbb{E}\|\xb_{L(S+T+1)}-\1_n\otimes\yb^\ast\|^2$ under various expected privacy $\nu$ and accuracy $\epsilon$.}
\label{fig-case-eq}
\end{figure}

%According to (\ref{eq:eq-T}), the minimal $T$ that guarantees the desired computation accuracy and the differential privacy with probability larger than $p$ are computed at various $L,\epsilon$, as shown in Table 7. Moreover, with these minimal $T$, using the MCM (with $10^5$ samples) the corresponding probabilities of the differential privacy are also simulated in Table 7. From this table, it is clear that for all $L=205, 258, 310$ and $\epsilon=0.001,0.01,0.1$, the derived minimal $T$ according to (\ref{eq:eq-T}) are  larger than the values in Figures \ref{fig-case-eq-0}, \ref{fig-case-eq-1} and \ref{fig-case-eq-2} to reach the corresponding expected accuracy, respectively. This in turn shows that the parameters $S,L,T$ according to (\ref{eq:eq-L})-(\ref{eq:eq-T}) can achieve the desired computation accuracy. Taking the failure probability of the differential privacy into account, as shown in Table 7, the resulting probabilities using the MCM are smaller than the corresponding expected probabilities $\rho$, respectively, i.e., the desired differential privacy is achieved following the design in (\ref{eq:eq-L})-(\ref{eq:eq-T}).
%

\begin{table}
    \centering
    \caption{A comparison of minimal $T$ to reach various accuracy $\nu$ in Theorem \ref{theorem-Equ} and Figure \ref{fig-case-eq}}
        \begin{subtable}[t]{2in}
        \centering
          \begin{tabular}{ccc}
    \hline
     $\nu$ &  Theo. \ref{theorem-Equ} &  Fig. \ref{fig-case-eq}\\
    \hline
  $10^{0}$ & 552 & 204 \\
  $10^{-1}$ & 579 & 217 \\
  $10^{-2}$ & 605 & 229 \\
  \hline
\end{tabular}
        \caption{$\epsilon=10^{-3}$}
    \end{subtable}
  \qquad
    \begin{subtable}[t]{2in}
    \centering
    \begin{tabular}{ccc}
    \hline
      $\nu$ &  Theo. \ref{theorem-Equ} &  Fig. \ref{fig-case-eq}\\
    \hline
  $10^{0}$ & 507 & 180\\
  $10^{-1}$ & 533 & 192 \\
  $10^{-2}$ & 559 & 207 \\
  \hline
\end{tabular}
        \caption{$\epsilon=10^{-2}$}
    \end{subtable}
    \qquad
    \begin{subtable}[t]{2in}
    \centering
        \begin{tabular}{ccc}
            \hline
      $\nu$ &  Theo. \ref{theorem-Equ} &  Fig. \ref{fig-case-eq}\\
    \hline
  $10^{0}$ & 461 & 155\\
  $10^{-1}$ & 487 & 170\\
  $10^{-2}$ & 513 & 182\\
  \hline
        \end{tabular}
        \caption{$\epsilon=10^{-1}$}
    \end{subtable}
    \label{tab:T-eq}
\end{table}

\begin{table}
    \centering
    \caption{Resulting differential privacy probability under various accuracy and privacy levels}
        \begin{subtable}[t]{2in}
        \centering
          \begin{tabular}{cc}
    \hline
     $\nu$ &    probability\\
    \hline
  $10^{0}$ &  0.9975 \\
  $10^{-1}$ &  0.9975 \\
  $10^{-2}$ &  0.9970 \\
  \hline
\end{tabular}
        \caption{$\epsilon=10^{-3}$}
    \end{subtable}
  \qquad
    \begin{subtable}[t]{2in}
    \centering
    \begin{tabular}{cc}
    \hline
      $\nu$ &   probability\\
    \hline
  $10^{0}$ & 0.9983\\
  $10^{-1}$ & 0.9970 \\
  $10^{-2}$ & 0.9971 \\
  \hline
\end{tabular}
        \caption{$\epsilon=10^{-2}$}
    \end{subtable}
    \qquad
    \begin{subtable}[t]{2in}
    \centering
        \begin{tabular}{cc}
            \hline
      $\nu$ &   probability\\
    \hline
  $10^{0}$ & 0.9977\\
  $10^{-1}$ & 0.9970\\
  $10^{-2}$ & 0.9967\\
  \hline
        \end{tabular}
        \caption{$\epsilon=10^{-1}$}
    \end{subtable}
    \label{tab:prob-eq}
\end{table}

\subsection{Classification with Logistic Loss Function}

Finally, we investigate the practical performance of the PPSC-Gossip-DCO algorithm in contrast to existing privacy-preserving distributed optimization protocols \cite{huang2015differentially,nozari2016differentially} for a linear classification problem. We assume each node $i$ of the public network $\mathrm{G}$ holds a database that consists of samples $(a_{i,j},b_{i,j})\in\R^m\times\{0,1\}$ with $j=1,\dots,N_i$ and $i=1,\ldots,n$. Note that $a_{i,j},b_{i,j}$ represent a sample's features and label, respectively. Then we specify the objective in (\ref{eq:opti}) as the training goal of the well-known logistic regression classifier, i.e.,
\begin{equation}\notag
f_i(\yb) = \sum\limits_{j=1}^{N_i} b_{ij}a_{ij}^\top\yb - \log(1+\exp(a_{ij}^\top\yb)) + \frac{\lambda}{2n}\|\yb\|^2\,,\quad i=1,\ldots,n.
\end{equation}

To conduct numerical experiments, we specify the feasible space $\mathtt{C}$ in (\ref{eq:opti}) as a unit ball and adopt the MNIST\footnote{See \url{http://yann.lecun.com/exdb/mnist/}} dataset as the overall database. MNIST is a database of handwritten digits consisting of $60000$ samples and $784$ features in the training set, and $10000$ samples in the test set. We artificially and evenly allocate all training samples to all nodes of $\mathrm{G}$ in Figure \ref{fig-ring-a}, so that all three protocols, including the PPSC-Gossip-DCO algorithm with the private network $\mathrm{G}_{\mathrm{p}}$ in Figure \ref{fig-G-p-1}, Huang2015 \cite{huang2015differentially}, Nozari2018 \cite{nozari2016differentially}, can be executed over $\mathrm{G}$ with uniform privacy budgets $\epsilon=0.001,0.01,0.1$ for $L=3000$ steps. At each $l$, we let $\bar{\yb}(l)=\sum\limits_{i=1}^n \yb_i(l) / n$ denote the network's estimate towards the global logistic model, based on which prediction is performed over both the training set and the test set. The Area under the Curve (AUC), calculated based on the true and the predicted labels, is a metric ranging from zero to one for model evaluation. It represents the probability that a random positive sample whose $b_{ij}=1$ has a larger predicted label than that of a random negative sample whose $b_{ij}=0$. Note that the larger the AUC is, the better prediction can be made under the model. A model with AUC 0.5 is no better than a random guess. In this example, we will evaluate the instantaneous models at each $l$ trained under the mentioned three algorithms by calculating their AUC, and plot the AUC trajectory for  various $\epsilon$ in Figure \ref{fig-auc}.

We observe in Figure \ref{fig-auc-1-0} - \ref{fig-auc-3-0}, although the AUC of Nozari2018 \cite{nozari2016differentially} promptly converges to $0.78,0.80,0.82$ for $\epsilon=0.001,0.01,0.1$, it soon gets overtaken by that of the PPSC-Gossip-DCO algorithm between the $1200$- and $1400$-th iteration. From the $1500$-th iteration, the AUC of the PPSC-Gossip-DCO algorithm stays at a high and steady level of $0.85$, implying that the model trained under the PPSC-Gossip-DCO algorithm has a much better predictive ability than that of Nozari2018 \cite{nozari2016differentially}, especially for the smallest $\epsilon=0.001$. In fact, thank to the ``shuffle and average" mechanism of the PPSC-Gossip-DCO algorithm, it manages to train almost the same and accurate model as the privacy-preserving requirement goes stronger. In addition, Huang2015 \cite{huang2015differentially}, which preserves the data privacy at a high cost of perturbing the shared states, only drives the AUC to a low level below $0.6$.

Based on the observations above, on the one hand, we find that the PPSC-Gossip-DCO algorithm outperforms the objective-perturbing method Nozari2018 \cite{nozari2016differentially} in the long run because of the exact objective function adopted in the PPSC-Gossip-DCO algorithm. This superiority is even more significant under high differential privacy requirements. On the other hand, the state-perturbing strategy Huang2015 \cite{huang2015differentially} gains the model nearly no knowledge, because exponentially decaying step size is implemented for the compliance of differential privacy, which provides no feasibility guarantee. In conclusion, for this experiment the PPSC-Gossip-DCO algorithm has an overwhelming advantage over Huang2015 \cite{huang2015differentially} and Nozari2018 \cite{nozari2016differentially} in terms of accuracy, at the price of, of course, a higher computation overhead.

\begin{figure}[ht]
	\centering
	\qquad \qquad \qquad
	\begin{subfigure}{.4\textwidth}
		\centering
		\includegraphics[width=1.15\linewidth]{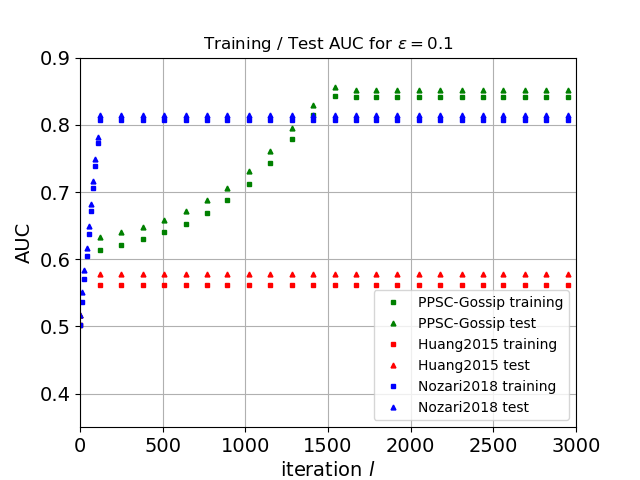}
		\caption{Expected accuracy $\nu=10^0$}
		\label{fig-auc-1-0}
	\end{subfigure}
	\newline
	\begin{subfigure}{.4\textwidth}
		\centering
		\includegraphics[width=1.15\linewidth]{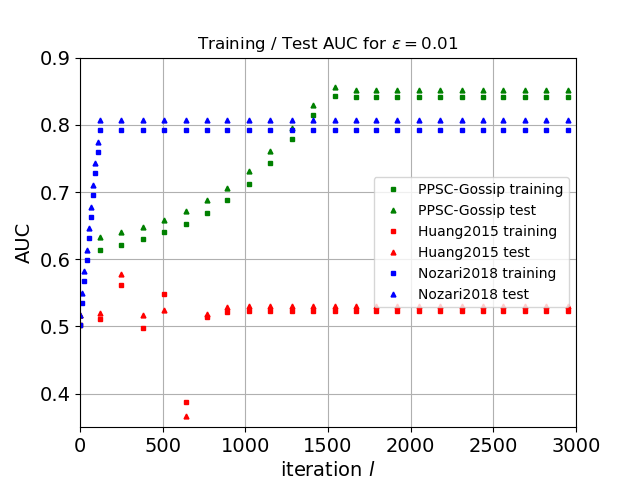}
		\caption{Expected accuracy $\nu=10^{-1}$}
		\label{fig-auc-2-0}
	\end{subfigure}
	\qquad \qquad
	\begin{subfigure}{.4\textwidth}
		\centering
		\includegraphics[width=1.15\linewidth]{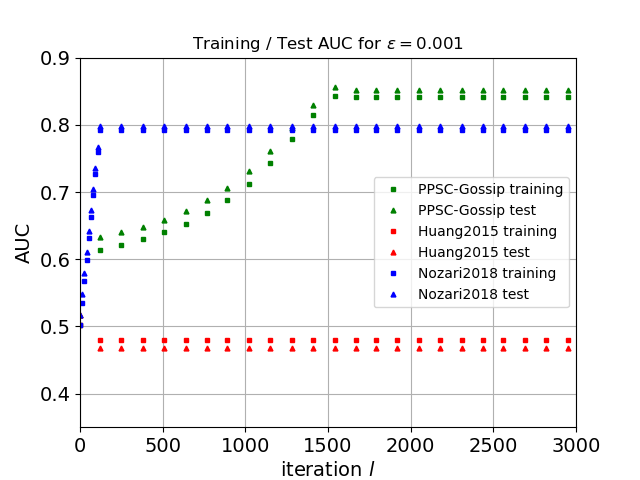}
		\caption{Expected accuracy $\nu=10^{-2}$}
		\label{fig-auc-3-0}
	\end{subfigure}
	\caption{The AUC under various expected privacy $\nu$ and accuracy $\epsilon$.}
	\label{fig-auc}
\end{figure}

\section{Conclusions}

In this paper, differentially private distributed computing protocols were developed for computation tasks of average consensus, network linear algebraic equation and distributed convex optimization via public-private networks. Such protocols were established on the Multi-Gossiping PPSC mechanism over a private communication graph, where randomly selected node pairs update their states in such a way that they are shuffled with random noise while maintaining summation consistency. By embedding the multi-gossip  PPSC mechanism for privacy encryption before an averaging consensus algorithm over the public network for local computations, we developed a PPSC-Gossip averaging consensus algorithm, which  can compute the average with any desired accuracy, while achieving any desired differential privacy with any prescribed probability. This PPSC-Gossip averaging consensus algorithm was then employed to develop two respective distributed computing protocols for network linear equations and distributed convex optimization. It was shown that both protocols can offer both differential privacy and computation accuracy guarantees.
In future works, it is of interest to investigate the possibility of integrating the PPSC framework or its extensions to distributed optimization problems without gradient information or problems with online settings where data arrives at the nodes sequentially.

%we plan to extend the idea of this paper to  more general distributed optimization problems, e.g., problems with non-convex constraints and cost functions, and to other distributed computing architectures such as distributed ADMM.

%This paper introduced a new PPSC-Gossip-based distributed computing framework with differential privacy guarantee. In our framework, the key idea is  to run the PPSC Gossip mechanism over the private network such that the node states are reshuffled with random noise while maintaining the summation of node states. By implementing this PPSC gossip mechanism as a basic privacy encryption subroutine, we concretely show that the proposed framework can be applied to solve the computation tasks, such as average consensus, network linear algebraic equations and distributed convex optimization, providing arbitrary both differential privacy and computation accuracy guarantees. Concrete case studies show that the proposed framework indeed provided improved computational  accuracy under the same privacy budget compared to existing work, and more importantly, such accuracy guarantee does not depend on the specific  level of privacy budget.

\appendix
\section{Proof of Theorem \ref{theorem-PPSC}}
\label{app:sec-ppsc}

\subsection{Preliminaries}

Without loss of generality, we let all nodes be indexed such that $i<j$ for any nodes $i\in\V^k$ and $j\in\V^{k+1}$.

Let $\V^k=\{k_1,\ldots,k_{n_k}\}$, and denote by $\mathsf{e}_{\rm p}^k(s)=\big(k_{s_{\rm out}},k_{s_{\rm in}}\big)$  the randomly selected edge  during gossiping procedure in component $\mathrm{G}_{\rm p}^k$ at $s\in[1,S]$. Let  $\mathsf{C}_{\mathsf{E}_{\rm p}^k}=\displaystyle\prod_{s=1}^{S}\mathsf{C}_{S-s+1}^k$ with $\mathsf{C}_{s}^k\in\R^{n_k\times n_k}$ be a matrix each column $h$ of which is $\eb_h$ except the $s_{\rm out}$-th column that is $\eb_{s_{\rm in}}$, and $\mathsf{D}_{\mathsf{E}_{\rm p}^k}\in\R^{n_k\times S}$ be a matrix the $h$-th column of which is $\prod_{s=1}^{S-h}\mathsf{C}_{S-s+1}^k(\eb_{h_{\rm out}}-\eb_{h_{\rm in}})$ for $h=1,2,\ldots,S-1$ and $\eb_{S_{\rm out}}-\eb_{S_{\rm in}}$ for $h=S$.
We note that both random matrices $\mathsf{C}_{\mathsf{E}_{\rm p}^k}$ and $\mathsf{D}_{\mathsf{E}_{\rm p}^k}$ are  determined by the randomly selected edge sequence $\mathsf{E}_{\rm p}^k$, and given any positive integer $S$, there is a finite-number set $\Theta_a^k$ of $\mathsf{E}_{\rm p}^k$. Thus, the map $\mathcal{M}_{\rm S}$ can be explicitly expressed by
\beeq{\label{eq:ppsc_dynamics}
\yb_{\rm ppsc} = \mathcal{M}_{\rm S}(\betab):=\mathsf{C}_{\mathsf{E}_{\rm p}} \betab + \mathsf{D}_{\mathsf{E}_{\rm p}} \gammab\,
}
with $\mathsf{C}_{\mathsf{E}_{\rm p}}=\blkdiag\big(\mathsf{C}_{\mathsf{E}_{\rm p}^1},\ldots,\mathsf{C}_{\mathsf{E}_{\rm p}^q}\big)$ and $\mathsf{D}_{\mathsf{E}_{\rm p}}=\blkdiag\big(\mathsf{D}_{\mathsf{E}_{\rm p}^1},\ldots,\mathsf{D}_{\mathsf{E}_{\rm p}^q}\big)$. Moreover, the expression of $\lambda_{\rm ppsc}$ can be simplified as
\[
\lambda_{\rm ppsc} = \min\{\lambda_{\rm ppsc}^1,\ldots,\lambda_{\rm ppsc}^q\} \,,\qquad \mbox{with $\lambda_{\rm ppsc}^k=\displaystyle\min_{\mathsf{E}_{\rm p}^k\in\Theta_a^k} \sigma_{m}^+\big(\mathsf{D}_{\mathsf{E}_{\rm p}^k}\big)$, $k=1,\ldots,q$.}
\]
%We remark that when each component $\mathsf{E}_{\rm p}^k$ of $\mathrm{G}_{\rm p}$  is comprised of two or three nodes, there holds $\lambda_a\geq 1$ for all $S$.

\subsection{A technical lemma}

Next, we present a lemma on $p_0$-privacy covering time for the multi-gossiping PPSC mechanism.
	
\begin{definition} \label{def:rand_encryption_time}
Introduce $\mathcal{Q}_s$ as the event that all nodes have altered their states at least once during the time $s\in[1,t]$.
For any $p_0\in(0,1)$, the $p_0$-privacy covering time for the multi-gossiping PPSC mechanism is defined by
\begin{equation}\notag
S_{p_0}(\mG_{\rm p})=\inf\{t: 1-\mathbb{P}(\mathcal{Q}_t)\le{p_0}\}.
\end{equation}
\end{definition}
	
In this definition, the $S_{p_0}(\mG_{\rm p})$ denotes the time needed to guarantee the event $\mathcal{Q}_s$ with probability $1-p_0$. Regarding this time $S_{p_0}(\mG_{\rm p})$, we present the following result.

\begin{lemma}\label{lemma-2}
For any ${p_0}\in(0,1)$, the ${p_0}$-privacy covering time associated with graph $\mG_{\rm p}$ for the multi-gossiping PPSC mechanism satisfies
\begin{equation}\notag
S_{p_0}(\mG_{\rm p}) \geq \frac{\log(1-(1-p_0)^{1\over q})-\log n_{\max}}{\log(1-\frac{1+r^\dag}{n_{\max}})}.
\end{equation}
\end{lemma}
\begin{proof}
For any $k=1,\ldots,q$, we denote $\mathcal{F}_{k_i}^t$  the event that node $k_i\in{\rm G}_{\rm p}^k$ has its state changed at least once during the time interval $[1, t]$, i.e., $\mathcal{F}_{k_i}^t=\{x_{k_i}(t) \neq x_{k_i}(0)\}$. Then, it is seen that node $k_i$ has its state changed at any $s\in\{1,\dots,t\}$ with $\frac{1}{n_k}+\frac{1}{n_k}\sum_{j=1}^{r_{k_i}}\frac{1}{r_{k_j}}$ probability.
%%	The event $\mcalQ_t$ holding true implies that the entire network states have been encoded by the R-PPSC-Gossip  algorithm.
	Then it follows
	\begin{equation}\label{eq:rand_xi_neq}
	\mathbb{P}(\mathcal{F}_{k_i}^{t}) = 1-\big(1-\frac{1}{n_k}-\frac{1}{n_k}\sum_{j=1}^{r_{k_i}}\frac{1}{r_{k_j}}\big)^t\geq 1-\big(1-\frac{1+r_k^\dag}{n_k}\big)^t
	\end{equation}
with $r_k^\dag=\frac{r_{k}^{\rm min}}{r_{k}^{\rm max}}$.
	Clearly, $\mcalQ_s=\bigcap\limits_{i\in\mV} \mathcal{F}_i^t$. By (\ref{eq:rand_xi_neq}) and the Fr{\'e}chet inequalities \cite{frechet1935generalisation}, we have
	\begin{equation}\label{eq:rand_encryption_time_lowerbound}
	\ba{rcl}
	\mathbb{P}(\mcalQ_t) &=& \prod\limits_{k=1}^{q}\mathbb{P}\left(\bigcap\limits_{i\in\V^k} \mathcal{F}_i^t\right) \ge \prod\limits_{k=1}^{q}\left(\sum\limits_{i=1}^{n_k} \mathbb{P}(\mathcal{F}_{k_i}^{t})-(n_k-1)\right)\,\\
&\geq& \prod\limits_{k=1}^{q}\left(1-{n_k}(1-\frac{1+r_k^\dag}{n_k})^t\right)\ge \left(1-{n_{\max}}\big(1-\frac{1+r^\dag}{n_{\max}}\big)^t\right)^q.
	\ea
	\end{equation}
	The proof is thus completed with (\ref{eq:rand_encryption_time_lowerbound}).
\end{proof}

\subsection{Proof of Theorem \ref{theorem-PPSC}}

From the structure of the  multi-gossip algorithm, it follows that for all possible edge sequences $\mathsf{E}_{\rm p}=\big\{\mathsf{E}_{\rm p}^1,\ldots,\mathsf{E}_{\rm p}^q\big\}$, we have $\|\mathsf{C}_{\mathsf{E}_{\rm p}}\|_1 =1$ and $\rank(\mathsf{D}_{\mathsf{E}_{\rm p}}) \leq n-q$.
Let $\mathsf{E}_{\rm p}^\ast$ be the sequence of communication edges observed by  eavesdroppers, and $\Cb\in\R^{n\times n}$ and $\Db\in\R^{n\times qS}$ be the resulting values of $\mathsf{C}_{\mathsf{E}_{\rm p}^\ast}$ and $\mathsf{D}_{\mathsf{E}_{\rm p}^\ast}$, respectively. Thus, the $\mathscr{M}_{\rm S}$ conditioned on $\mathsf{E}_{\rm p}^\ast$ takes the form
\beeq{\label{eq:Mech-ave}
\yb_{\rm ppsc} = \mathscr{M}_{\rm S}\big({\betab} \,|\, \mathsf{E}_{\rm p}^\ast\big) = \Cb \betab + \Db\gammab\,
}
with $\|\Cb\|_1 =1$ and $\rank(\Db) \leq n-q$.
Then, we denote $\rank(\Db) := r \leq n-q$, and $\Db^{\perp}\in\R^{(n-r)\times n}$ be such that $\Db^{\perp}\Db=0$ and $\rank(\Db^{\perp})=n-r$.
Let $U_a\in\mathbb{R}^{r\times n}$ be such that $U_aU_a^\top = \Ib_r$ and $\mathbf{T}_{\rm M}:=\begin{bmatrix}\Db^{\perp} \cr U_a\end{bmatrix}\in\mathbb{R}^{n\times n}$ is nonsingular. Thus, the matrix $\Sigma_a:=U_a\Db$ satisfies $\rank(\Sigma_a)=r$ and $\sigma_m(\Sigma) = \lambda_{\rm ppsc}^2$ with $\Sigma:=\Sigma_a\Sigma_a^\top$.

With  $p_0=1-\rho$ and $S \geq S_{\rho}^\ast$, it can be seen from Lemma \ref{lemma-2} that the event $\mathcal{Q}_S$ occurs with the probability larger than $\rho$.
Conditioned on $\mathcal{Q}_S$, we proceed to analyze the differential privacy  of 	the mechanism (\ref{eq:Mech-ave}).
According to \cite{Ny-Pappas-TAC-2014}, for any two $\mu$-adjacent $\betab,{\betab}^\prime\in\R^n$,   we have
\[\ba{rcl}
&&\mathbb{P}\big(\mathscr{M}_{\rm S}(\betab \,|\, \mathsf{E}_{\rm p}^\ast)\in R\big)
=\mathbb{P}\big(\mathbf{T}_{\rm M}\Cb \betab + \mathbf{T}_{\rm M}\Db\gammab\in \mathbf{T}_{\rm M}R\big) \,\\
&=& \mathbb{P}(U_a\Cb \betab + \Sigma_a\gammab\in R')\,\\
&=&(2\pi)^{-\frac{r}{2}}\frac{1}{\sigma_{\gammab}\det(\Sigma)^{\frac{1}{2}}} \displaystyle\int_{R'}\exp\bigg(-\frac{1}{2\sigma_{\gammab}^2}(u-U_a\Cb \betab)^\top\Sigma^{-1} (u-U_a\Cb \betab)\bigg){\rm d} u\,\\
&\leq& \exp(\epsilon)\cdot\mathbb{P}(U_a\Cb \betab^\prime + \Sigma_a\gammab\in R')  + \mathbb{P}\bigg( \frac{1}{\sigma_{\gammab}}{\tilde\betab}^\top \Cb^\top U_a^\top \Sigma^{-\frac{1}{2}}v \geq \epsilon- \frac{1}{2\sigma_{\gammab}^2}\|\Sigma^{-\frac{1}{2}}U_a\Cb\tilde\betab\|^2\bigg)\,\\
&\leq& \exp(\epsilon)\cdot\mathbb{P}\big(\mathbf{T}_{\rm M}\Cb \betab^\prime + \mathbf{T}_{\rm M}\Db\gammab\in \mathbf{T}_{\rm M}R\big)  + \mathbb{P}\bigg( Z \geq \frac{\epsilon\sigma_{\gammab}\sqrt{\sigma_m(\Sigma)}}{\|\Cb\|_{1}\mu}- \frac{\|\Cb\|_{1}\mu}{2\sigma_{\gammab}\sqrt{\sigma_m(\Sigma)}}\bigg)\,\\
&=& \exp(\epsilon)\cdot\mathbb{P}\big(\mathscr{M}_{\rm S}(\betab^\prime \,|\, \mathsf{E}_{\rm p}^\ast)\in R\big)  + \mathbb{P}\bigg( Z \geq \frac{\epsilon\sigma_{\gammab}\sqrt{\sigma_m(\Sigma)}}{\|\Cb\|_{1}\mu}- \frac{\|\Cb\|_{1}\mu}{2\sigma_{\gammab}\sqrt{\sigma_m(\Sigma)}}\bigg)
\ea\]	
with  $R^\prime=U_a R$,  $\tilde{\betab}={\betab}-{\betab}^\prime$,
$v\sim \mathcal{N}(0,\Ib_{r})$ and $Z\sim \mathcal{N}(0,1)$.	
	
Thus, the $(\epsilon,\delta)$-differential privacy is preserved with probability larger than $\rho$, if
\beeq{\label{eq:delta}
\delta \geq \mathcal{Q}\bigg(\frac{\epsilon\sigma_{\gammab}\sqrt{\sigma_m(\Sigma)}}{\|\Cb\|_{1}\mu}- \frac{\|\Cb\|_{1}\mu}{2\sigma_{\gammab}\sqrt{\sigma_m(\Sigma)}}\bigg)\,\,.
}
With $\mathcal{Q}(w)$ being a strictly decreasing smooth function, $\sigma_m(\Sigma)=\lambda_{\rm ppsc}^2$ and $\|\Cb\|_{1}=1$, it is clear that (\ref{eq:delta}) is equivalent to
\[
 \frac{\epsilon\lambda_{\rm ppsc}}{\mu}\sigma_{\gammab}^2 - \mathcal{Q}^{-1}(\delta) \sigma_{\gammab}- \frac{\mu}{2\lambda_{\rm ppsc}} \geq 0\,,
\]
which is clearly true by recalling that $\sigma_\gammab \geq \frac{{\mu \kappa(\epsilon,\delta)}}{\lambda_{\rm ppsc}}$.

\section{Proof of Theorem \ref{theorem-Ave}}
\label{app:sec-1}

Since no extra randomness is introduced  during the averaging consensus stage, it can be seen that the desired $(\epsilon,\delta)$-differential privacy is guaranteed with probability higher than $\rho$  by the proof of Theorem \ref{theorem-PPSC} in Appendix \ref{app:sec-ppsc}.
In the following, we will focus on the computation accuracy.

\begin{lemma}\label{lemma-B}
Along the PPSC-Gossip-AC algorithm, there holds
  \beeq{\label{eq:ave-error}
\mathbb{E}\|\xb_{S+T}- \frac{1}{n}{\bf 1}_n{\bf 1}_n^\top\xb_0\|^2 \leq [n\|\xb_0\|^2 + 2q^2S^2\sigma_{\gammab}^2] \cdot (1-\lambda_{\rm G})^{2T}.
}
\end{lemma}
\begin{proof}

We first observe that
\beeq{\label{eq:iii_1}
\mathbb{E}\,\|\xb_S- \frac{1}{n}{\bf 1}_n{\bf 1}_n^\top\xb_0\|^2=\Exp\,\|\xb_S\|^2 + \frac{1}{n}\left|\1_n^\top\xb_0\right|^2 - \frac{2}{n}\1_n^\top\xb_0\1_n^\top\Exp\xb_S.
}
Next,  we will seek the upper bound of each term on the right hand side of (\ref{eq:iii_1}) so as to get the upper bound of $\mathbb{E}\,\|\xb_S- \frac{1}{n}{\bf 1}_n{\bf 1}_n^\top\xb_0\|^2$. By (\ref{eq:ppsc_dynamics}) and $\Exp\,\gammab=0$, the first term satisfies
\beeq{\label{eq:iii_2}
\Exp \,\|\xb_S\|^2
\le \Exp\, \|\mathsf{C}_{\mathsf{E}_{\rm p}}\|_{\rm F}^2\|\xb_0\|^2 + \Exp\,\|\mathsf{D}_{\mathsf{E}_{\rm p}}\|_{\rm F}^2\|\gammab\|^2.
}
By the structure of the PPSC multi-gossiping mechanism, it can be verified that there are deterministically $n$ ones in the entries of $\mathsf{C}_{\mathsf{E}_{\rm p}}$, and at most $qS$ ones and $qS$ minus ones in $\mathsf{D}_{\mathsf{E}_{\rm p}}$, i.e., $\|\mathsf{C}_{\mathsf{E}_{\rm p}}\|_{\rm F}=\sqrt{n},\|\mathsf{D}_{\mathsf{E}_{\rm p}}\|_{\rm F}\leq \sqrt{2qS}$. In addition, $\Exp\,\|\gammab\|^2 =qS\sigma_{\gammab}^2$. Then it follows from (\ref{eq:iii_2}) that
\beeq{
\Exp\, \|\xb_S\|^2 \le n\|\xb_0\|^2 + 2q^2S^2\sigma_{\gammab}^2.\label{eq:iii_3}
}
To analyze the third term of (\ref{eq:iii_1}), we observe that $\xb_S=\mathsf{C}_{\mathsf{E}_{\rm p}} \xb_0 + \mathsf{D}_{\mathsf{E}_{\rm p}} \gammab$. By the structure of the PPSC multi-gossiping mechanism again, we note that there hold
$\1_n^\top\mathsf{C}_{\mathsf{E}_{\rm p}} = \1_n^\top$ and $\1_n^\top\mathsf{D}_{\mathsf{E}_{\rm p}}=0$, deterministically. Thus, we have
\beeq{\label{eq:iii_4}\ba{rcl}
\frac{2}{n}\1_n^\top\xb_0\1_n^\top\Exp\xb_S &=& \frac{2}{n}\1_n^\top\xb_0\big(\Exp(\1_n^\top\mathsf{C}_{\mathsf{E}_{\rm p}}) \xb_0 + \Exp\1_n^\top\mathsf{D}_{\mathsf{E}_{\rm p}}\gammab\big)\,\\
&=& \frac{2}{n}|\1_n^\top \xb_0|^2
\ea}
Thus, substituting (\ref{eq:iii_3}) and (\ref{eq:iii_4}) into (\ref{eq:iii_1}) yields
\beeq{\ba{rcl}\label{eq:iii_5}
\mathbb{E}\,\|\xb_S- \frac{1}{n}{\bf 1}_n{\bf 1}_n^\top{\xb_0}\|^2 &\le& n \|\xb_0\|^2 + 2q^2S^2\sigma_{\gammab}^2 - \frac{1}{n}\left|\1_n^\top\xb_0\right|^2 \,
\leq n \|\xb_0\|^2 + 2q^2S^2\sigma_{\gammab}^2\,.
\ea}

With the upper bound in (\ref{eq:iii_5}), we now proceed to show the bound of $\mathbb{E}\left(\|\xb_{S+T}- \frac{1}{n}{\bf 1}_n{\bf 1}_n^\top\xb_0\|^2\right)$. Notice that $\xb_{s+1} = (\Ib_n-\Ab)\xb_{s}$ for $s\geq S$, which yields
\[\ba{rcl}
\|\xb_{S+T}- \frac{1}{n}{\bf 1}_n{\bf 1}_n^\top\xb_0\|^2 &=&\|(\Ib_n-\Ab-\frac{1}{n}{\bf 1}_n{\bf 1}_n^\top)(\xb_{S}- \frac{1}{n}{\bf 1}_n{\bf 1}_n^\top\xb_0)\|^2  \leq \|\xb_S- \frac{1}{n}{\bf 1}_n{\bf 1}_n^\top\xb_0\|^2(1-\lambda_{\rm G})^{2T}\,.
\ea\]
This, together with (\ref{eq:iii_5}), completes the proof.
\end{proof}

This lemma in turn proves the desired computation accuracy in combination with $\xb_0=\db$ and (\ref{eq:theo2-T}).

%%%%%%%%%%%%%%%%%%%%%%%%%%%%%%%%%%%%%%%%%%%%%%%%%%%%%%%%%
%%%%%%%%%%%%%%%%%%%%%%%%%%%%%%%%%%%%%%%%%%%%%%%%%%%%%%%%
\section{Proof of Theorem \ref{theorem-Equ}}
\label{app:equ}
%%%

\subsection{Preliminaries}

Let $\zetab_{0}= \1_n\otimes \zeta_0$ and $\zetab_{l}=\xb_{(l-1)(S+T+1)+S}$, $l=1,\ldots,L$.
Denote by $\gamma_k(s)\in\mathbb{R}^m$ the randomly and independently generated noise vector in the $k$-th component of ${\rm G}_{\rm p}$ at time $s\in[l(S+T+1)+1, l(S+T+1)+S]$, $l=0,1,\ldots,L-1$.
Let  $\gammab_{l}=[\gammab_{1,l};\gammab_{2,l};\ldots;\gammab_{q,l}]$ with
$$
\gammab_{k,l}=[\gammab_k(l(S+T+1)+1);\gammab_k({l(S+T+1)+2});\ldots;\gammab_k(l(S+T+1)+S)]\,,
$$
for $l=0,1,\ldots,L-1$.
For any $\xb=[x_1;\ldots;x_n]$ with $x_k\in\R^m$, $k=1,\ldots,n$, we denote
  \[
  \mathpzc{P}_{\Eq}(x_1,\ldots,x_n)= \begin{pmatrix}
                                       \mathpzc{P}_{\Eq_1}(x_1) \\
                                       \vdots \\
                                       \mathpzc{P}_{\Eq_n}(x_n)
                                     \end{pmatrix}
                                     := (\Ib_{mn} - \mathcal{H}_{\rm H})\xb + \mathcal{Z}_{\rm H}
  \]
with
  \[
  \mathcal{H}_{\rm H} = \diag\left(\frac{\Hb_1\Hb_1^{\top}}{\Hb_1^{\top}\Hb_1}\,\ldots\,\frac{\Hb_n\Hb_n^{\top}}{\Hb_n^{\top}\Hb_n}\right) \,,\quad
\mathcal{Z}_H = \begin{pmatrix}
  \frac{z_1\Hb_1^{\top}}{\Hb_1^{\top}\Hb_1} &
  \ldots &
  \frac{z_n\Hb_n^{\top}}{\Hb_n^{\top}\Hb_n}
\end{pmatrix}^{\top}\,.
  \]
According to Algorithm 4, we  observe that
\[\ba{rl}
  \zetab_{l+1} &= \mathsf{C}_{\mathcal{E}_{{\rm p},l}} \mathpzc{P}_{\Eq}\big(((\Ib_n-\Ab)^T\otimes \Ib_m)\zetab_{l}\big) + \mathsf{D}_{\mathcal{E}_{{\rm p},l}}\gammab_{l} \,\\
  &= \mathsf{C}_{\mathcal{E}_{{\rm p},l}}(\Ib_{mn} - \mathcal{H}_{\rm H})((\Ib_n-\Ab)^T\otimes \Ib_m)\zetab_{l} + \mathsf{C}_{\mathcal{E}_{{\rm p},l}} \mathcal{Z}_{\rm H} + \mathsf{D}_{\mathcal{E}_{{\rm p},l}}\gammab_{l}
  \ea
\]
for $l=0,1,\ldots,L-1$, where $\mathsf{C}_{\mathcal{E}_{{\rm p},l}}=\blkdiag\big(\mathsf{C}_{\mathsf{E}_{{\rm p},l}^1},\ldots,\mathsf{C}_{\mathsf{E}_{{\rm p},l}^q}\big)\otimes \Ib_m$ and $\mathsf{D}_{\mathcal{E}_{{\rm p},l}}=\blkdiag\big(\mathsf{D}_{\mathsf{E}_{{\rm p},l}^1},\ldots,\mathsf{D}_{\mathsf{E}_{{\rm p},l}^q}\big)\otimes \Ib_m$, with $\mathsf{C}_{\mathsf{E}_{{\rm p},l}^k}=\displaystyle\prod_{s=1}^{S}\mathsf{C}_{l(L+S+1)+S-s+1}^k$  and $\mathsf{D}_{\mathsf{E}_{{\rm p},l}^k}\in\mathbb{R}^{n_k\times S}$ being a matrix the $h$-th column of which is $\prod_{s=1}^{S-h}\mathsf{C}_{l(L+S+1)+S-s+1}^k(\eb_{h_{\rm out}}-\eb_{h_{\rm in}})$ for $h=1,2,\ldots,S-1$ and $\eb_{S_{\rm out}}-\eb_{S_{\rm in}}$ for $h=S$, in which $\mathsf{C}_s^k$ follows the definition in Appendix \ref{app:sec-ppsc}.

Denote $\phib_l = \xb_{l(S+T+1)}$, $l=0,1,\ldots,L$. Thus, we have
\[\ba{rcl}
\phib_{l+1} &=& \mathpzc{P}_{\Eq}\left(((\Ib_n-\Ab)^T\otimes \Ib_m)\zetab_{l+1}\right)\,\\
&=& (\Ib_{mn} - \mathcal{H}_{\rm H}) ((\Ib_n-\Ab)^T\otimes \Ib_m)\zetab_{l+1} + \mathcal{Z}_{\rm H}\,\\
&=& (\Ib_{mn} - \mathcal{H}_{\rm H}) ((\Ib_n-\Ab)^T\otimes \Ib_m)[\mathsf{C}_{\mathcal{E}_{{\rm p},l}} \phib_{l} + \mathsf{D}_{\mathcal{E}_{{\rm p},l}} \gammab_{l}] + \mathcal{Z}_{\rm H}\,\\
&=& \frac{1}{n}(\Ib_{mn} - \mathcal{H}_{\rm H})(\1_n\1_n^\top\otimes \Ib) \phib_{l} + \mathcal{Z}_{\rm H} + (\Ib_{mn} - \mathcal{H}_{\rm H})\Delta_{l}(\phib_{l})
\ea
\]
for $l=0,1,\ldots,L-1$, where $\phib_0=(\Ib_{mn} - \mathcal{H}_{\rm H}) ((\Ib_n-\Ab)^T\otimes \Ib_m)\zetab_{0} + \mathcal{Z}_{\rm H}$ and
\beeq{\label{eq:Delta_l}
\Delta_{l}(\phib_{l}) = ((\Ib_n-\Ab)^T\otimes \Ib_m)[\mathsf{C}_{\mathcal{E}_{{\rm p},l}} \phib_{l} + \mathsf{D}_{\mathcal{E}_{{\rm p},l}}\gammab_{l}] - \frac{1}{n}\big((\1_n\1_n^\top)\otimes \Ib_m\big) \phib_{l}\,.
}
Before we proceed to the proof of the theorem, the following lemmas are formulated.
\begin{lemma}
  \label{lemma-E}
  There holds
  \begin{equation}\label{eq:lambda_H}
  \big\|\frac{1}{n}(\Ib_{mn} - \mathcal{H}_{\rm H})(\1_n\1_n^\top\otimes \Ib_m)\big\| \leq \lambda_H\,.
\end{equation}
\end{lemma}
\begin{proof}
  We consider an auxiliary system
\[
X_{l+1} =  \frac{1}{n}(\Ib_{mn} - \mathcal{H}_{\rm H})(\1_n\1_n^\top\otimes \Ib_{m}) X_{l}\,,\quad X\in\mathbb{R}^{mn}\,.
\]
Letting $\vartheta_l=\frac{1}{n}(\1_n^\top\otimes \Ib) X_l$ yields
$
\vartheta_{l+1} = \big(\Ib_{mn} -\frac{1}{n}\sum_{i=1}^n \frac{\Hb_i\Hb_i^\top}{\Hb_i^\top\Hb_i}\big) \vartheta_{l}\,.
$
Thus we have $\|\vartheta_l\|\leq \|\vartheta_0\|\lambda_H^l$, which yields
\[
\|X_{l+1}\| = \big\|\frac{1}{n}(\Ib_{mn} - \mathcal{H}_{\rm H})(\1_n\1_n^\top\otimes \Ib_{mn}) X_{l}\big\| = \big\|(\Ib_{mn} - \mathcal{H}_{\rm H})(\1_n\otimes \Ib_{m}) \vartheta_{l}\big\| \leq \sqrt{n}\|\vartheta_0\|\lambda_H^l\,.
\]
Therefore, it can be inferred that $X_l$ exponentially converges to zero  and (\ref{eq:lambda_H}) must hold.
\end{proof}

\begin{lemma}\label{lemma-E-nu}
Given any $l=0,1,\ldots,L-1$, if $\|\Delta_{i}(\phib_{i})\|^2 \leq \nu$ holds for all $i=0,1,\ldots,l$, then there must hold $
\|\phib_{l+1}\| \leq \phi^\ast$, with $\phi^\ast=2\sqrt{n}\|\yb^\ast\|+ \sqrt{n}\|\zeta_0\|  + \frac{2-\lambda_H}{1-\lambda_H}\sqrt{\nu}$.
\end{lemma}
\begin{proof}
Note that
\beeq{\label{eq:phib-y}
\phib_{l+1} - \1_n\otimes\yb^\ast = \frac{1}{n}(\Ib_{mn} - \mathcal{H}_{\rm H})(\1_n\1_n^\top\otimes \Ib_m) (\phib_{l}- \1_n\otimes\yb^\ast)  + (\Ib_{mn} - \mathcal{H}_{\rm H})\Delta_{l}(\phib_{l})\,.
}
with $\phib_0-\1_n\otimes\yb^\ast=(\Ib_{mn} - \mathcal{H}_{\rm H}) ((\Ib_n-\Ab)^T\otimes \Ib_m)(\1_n\otimes(\zeta_0-\yb^\ast))$.
By Lemma \ref{lemma-E}, if $\|\Delta_{i}(\phib_{i})\|^2 \leq \nu$ holds for all $i=0,1,\ldots,l$, we then have $\|\phib_{l+1} - \1_n\otimes\yb^\ast\|\leq \lambda_H \|\phib_{l} - \1_n\otimes\yb^\ast\| + \sqrt{\nu}$, leading to
\[
\|\phib_{l+1} - \1_n\otimes\yb^\ast\| \leq \sqrt{n}\|\zeta_0 - \yb^\ast\| \lambda_H^{(l+1)} + \frac{2-\lambda_H}{1-\lambda_H}\sqrt{\nu}\,.
\]
Thus, we  have $\|\phib_{l+1}\| \leq \phi^\ast$, which completes the proof.
\end{proof}

\subsection{Proof of statement (i).}

Let $\mathcal{E}_{{\rm p},l}^\ast$ be the sequence of communication edges that is observed by the eavesdroppers during the time $s\in[l(S+L+1),l(S+L+1)+S]$, and ${\Cb}_{l}\in\R^{mn\times mn}$ and $\Db_{l}\in\R^{mn\times mS}$ be the resulting values of $\mathsf{C}_{\mathcal{E}_{{\rm p},l}^\ast}$ and $\mathsf{C}_{\mathcal{E}_{{\rm p},l}^\ast}$.
For any $0<p_0<1$, by the arguments of Lemma 1, it follows that if $S\geq \frac{\log{(1-p_0^{1\over q})}-\log n_{\max}}{\log(1-\frac{1+r^\dag}{n_{\max}})}$, the event $\mcalQ_S^l$ that all node states have altered during the time $s\in[l(S+T+1)+1,l(S+T+1)+S]$  occurs with probability larger than $p_0$ for $l=0,1,\ldots,L-1$, i.e., $\mathbb{P}(\mcalQ_s^l)\geq p_0$.
Then we let $p_0=\rho^{1\over 2L}$, and the following lemmas are formulated with $S,T$ being in Theorem \ref{theorem-Equ}.

\begin{lemma}
There holds
\[\ba{l}
\mathbb{P}\left(\|\Delta_{0}(\phib_{0})\|^2 \leq {\nu}\,|\,\mcalQ_S^{0}\right) \geq p_0\,\\
\mathbb{P}\left(\|\Delta_{l}(\phib_{l})\|^2 \leq {\nu}\,|\,\big(\bigcap_{i=0}^{l}\mcalQ_S^{i}\big) \bigcap \big(\bigcap_{i=0}^{l-1}\|\Delta_{i}(\phib_{i})) \|^2 \leq {\nu}\big)\right) \geq p_0\,,
\qquad\quad \mbox{for all }l=1,2,\ldots,L-1
\ea\]
\end{lemma}

\begin{proof}
Following the arguments of Lemma \ref{lemma-B}, it can be easily seen that
\[\ba{l}
\mathbb{E}\left(\|\Delta_{0}(\phib_{0})\|^2 \,|\,\mcalQ_S^{0}\right) \leq (n\|\phib_{0}\|^2 + 2q^2S^2\sigma_{\gammab}^2) (1-\lambda_{\rm G})^{2T}\leq (n{\phi^\ast}^2 + 2q^2S^2\sigma_{\gammab}^2) (1-\lambda_{\rm G})^{2T}.
\ea\]
Using Markov's inequality, we have $\mathbb{P}\left(\|\Delta_{0}(\phib_{0})\|^2 \geq {\nu}\right) \leq \frac{\mathbb{E}\left(\|\Delta_{l}(\phib_{l})\|^2 \right)}{\nu}$.
Thus, we have
\beeq{\label{eq:100}
\mathbb{P}\left(\|\Delta_{0}(\phib_{0})\|^2 \leq {\nu}\right) \geq 1-\frac{[n{\phi^\ast}^2 + 2q^2S^2\sigma_{\gammab}^2] (1-\lambda_{\rm G})^{2T}}{\nu}\,.
}

Similarly, we now prove the second inequality. With Lemma \ref{lemma-E-nu} and by the arguments of Lemma \ref{lemma-B}, we have
\[\ba{l}
\mathbb{E}\left(\|\Delta_{l}(\phib_{l})\|^2 \,|\,\big(\bigcap_{i=0}^{l-1}\mcalQ_S^{i}\big) \bigcap \big(\bigcap_{i=0}^{l-1}\|\Delta_{i}(\phib_{i})) \|^2 \leq {\nu}\big)\right) \leq (n{\phi^\ast}^2 + 2q^2S^2\sigma_{\gammab}^2) (1-\lambda_{\rm G})^{2T}\,
\ea\]
for all $l=1,2,\ldots,L-1$.
Thus, we have
\[\ba{rcl}
&&\mathbb{P}\left(\|\Delta_{l}(\phib_{l})\|^2 \leq {\nu}\,|\,\big(\bigcap_{i=0}^{l-1}\mcalQ_S^{i}\big) \bigcap \big(\bigcap_{i=0}^{l-1}\|\Delta_{i}(\phib_{i})) \|^2 \leq {\nu}\big)\right)\,\\
&=&1-\mathbb{P}\left(\|\Delta_{l}(\phib_{l})\|^2 \geq {\nu}\,|\,\big(\bigcap_{i=0}^{l-1}\mcalQ_S^{i}\big) \bigcap \big(\bigcap_{i=0}^{l-1}\|\Delta_{i}(\phib_{i})) \|^2 \leq {\nu}\big)\right)\,\\
&\geq& 1 - \displaystyle\frac{\mathbb{E}\left(\|\Delta_{l}(\phib_{l})\|^2  \,|\,\big(\bigcap_{i=0}^{l-1}\mcalQ_S^{i}\big) \bigcap \big(\bigcap_{i=0}^{l-1}\|\Delta_{i}(\phib_{i})) \|^2 \leq {\nu}\big)\right)}{\nu}\,.\\
&\geq& 1 - \displaystyle\frac{[n{\phi^\ast}^2 + 2q^2S^2\sigma_{\gammab}^2] (1-\lambda_{\rm G})^{2T}}{\nu}
\ea\]
This together with (\ref{eq:100}) completes the proof by combining with the fact that
\[
T\geq \frac{\log\big((1-p_0)\nu\big)-\log({\phi^\ast}^2 + 2q^2S^2\sigma_{\gammab}^2)}{2\log(1-\lambda_{\rm G})}\,.
\]
\end{proof}

With this lemma, we then can obtain that
\[
\mathbb{P}\left(\mcalQ_S^{0} \bigcap \|\Delta_{0}(\phib_{0})) \|^2  \leq {\nu}\right) = \mathbb{P}\left(\|\Delta_{0}(\phib_{0})\|^2 \leq {\nu}\,|\,\mcalQ_S^{0}\right) \mathbb{P}\left(\mcalQ_S^{0}\right) \geq (p_0)^{2}\,,
\]
and recursively for $l=1,\ldots,L-1$, we have
\[\ba{rcl}
&&\mathbb{P}\left(\bigcap_{i=0}^{l}\bigg(\mcalQ_S^{i} \bigcap \|\Delta_{i}(\phib_{i})) \|^2  \leq {\nu}\bigg)\right) \,\\
&=&  \mathbb{P}\left(\|\Delta_{l}(\phib_{l})\|^2 \leq {\nu}\,|\,\big(\bigcap_{i=0}^{l}\mcalQ_S^{i}\big) \bigcap \big(\bigcap_{i=0}^{l-1}\|\Delta_{i}(\phib_{i})) \|^2 \leq {\nu}\big)\right) \mathbb{P}\left(\big(\bigcap_{i=0}^{l}\mcalQ_S^{i}\big) \bigcap \big(\bigcap_{i=0}^{l-1}\|\Delta_{i}(\phib_{i})) \|^2 \leq {\nu}\big)\right)\,\\
&\geq& p_0\mathbb{P}\left(\big(\bigcap_{i=0}^{l}\mcalQ_S^{i}\big) \bigcap \big(\bigcap_{i=0}^{l-1}\|\Delta_{i}(\phib_{i})) \|^2 \leq {\nu}\big)\right)\,\\
&=& p_0\mathbb{P}\left(\mcalQ_S^{l}\right)\mathbb{P}\left(\bigcap_{i=0}^{l-1}\bigg(\mcalQ_S^{i} \bigcap \|\Delta_{i}(\phib_{i})) \|^2  \leq {\nu}\bigg)\right)\,\\
&\geq& (p_0)^2\cdot\mathbb{P}\left(\bigcap_{i=0}^{l-1}\bigg(\mcalQ_S^{i} \bigcap \|\Delta_{i}(\phib_{i})) \|^2  \leq {\nu}\bigg)\right)\,.
\ea\]
In view of the above analysis, we thus have
\[
\mathbb{P}\left(\bigcap_{l=0}^{L-1}\bigg(\mcalQ_S^{l} \bigcap \|\Delta_{l}(\phib_{l})) \|^2  \leq {\nu}\bigg)\right) \geq (p_0)^{2L}=\rho\,.
\]
Conditioned on the event $\bigcap_{l=0}^{L-1}\big(\mcalQ_S^{l} \bigcap \|\Delta_{l}(\phib_{l})) \|^2 \leq {\nu}\big)$, we are now proceeding to analyze the differential privacy of PPSC-Gossip network linear-equation solver.
Let $\mathcal{E}_{\rm p}^\ast=\big\{\mathcal{E}_{{\rm p},0}^\ast,\mathcal{E}_{{\rm p},1}^\ast,\ldots,\mathcal{E}_{{\rm p},L-1}^\ast\big\}$. For any two $\mu$-adjacent equations $\mathfrak{E}:\Hb\yb=\zb,\ \mathfrak{E}^\prime:\Hb^\prime\yb=\zb^\prime$, we let  $R_l\subseteq\mathbb{R}^{mn}$ and denote the events
\[\ba{rcl}
{\mathcal{F}_l^{\rm ea}}=\{\Cb_{l}(\Ib_{mn} - \mathcal{H}_{\rm H})((\Ib_n-\Ab)^T\otimes \Ib_m)\zetab_{l} + \Cb_{l} \mathcal{Z}_{\rm H} + \Db_{l}\gammab_{l}\in R_l\}\,\\
{\mathcal{F}_l^{\rm ea}}^\prime=\{\Cb_{l}(\Ib_{mn} - \mathcal{H}_{\rm H}^\prime)((\Ib_n-\Ab)^T\otimes \Ib_m)\zetab_{l} + \Cb_{l} \mathcal{Z}_{\rm H}^\prime + \Db_{l}\gammab_{l}\in R_l\}
\ea\]
for $l=0,1,\ldots,L-1$, with
  \[
  \mathcal{H}_{\rm H}^\prime = \diag\begin{pmatrix}\frac{\Hb_1^\prime{\Hb_1^\prime}^{\top}}{{\Hb_1^\prime}^{\top}\Hb_1^\prime}&\ldots&\frac{{\Hb_n^\prime}{\Hb_n^\prime}^{\top}}{{\Hb_n^\prime}^{\top}{\Hb_n^\prime}}\end{pmatrix} \,,\quad
\mathcal{Z}_H^\prime = \begin{pmatrix}
  \frac{z_1^\prime{\Hb_1^\prime}^{\top}}{{\Hb_1^\prime}^{\top}{\Hb_1^\prime}} &
  \ldots &
  \frac{z_n^\prime{\Hb_n^\prime}^{\top}}{{\Hb_n^\prime}^{\top}{\Hb_n^\prime}}
\end{pmatrix}^{\top}\,.
  \]

Thus, recalling the fact that  the Average-Consensus procedure in Algorithm 4 is deterministic, we can see that for all $R\subseteq\mbox{range}(\mathscr{M}_{\rm eq})$, there always exist $R_l\subseteq\mathbb{R}^{mn}$, $l=0,1,\ldots,L-1$ such that there holds
\[\ba{l}
\mathbb{P}\big(\mathscr{M}_{\rm eq}\big(\mathfrak{E} \,|\, \mathcal{E}_{\rm p}^\ast\big)\in R\big) = \mathbb{P}\big(\bigcap_{l=0}^{L-1}{\mathcal{F}_l^{\rm ea}}\big) = \mathbb{P}\big({\mathcal{F}_0^{\rm ea}}\big)\cdot\prod_{l=1}^{L-1}\mathbb{P}\big({\mathcal{F}_l^{\rm ea}}|\,\bigcap_{k=0}^{l-1}{\mathcal{F}_k^{\rm ea}}\big)\\
\ea\]	
With this in mind, we observe that
\[\ba{rcl}
&&\|\big(\Cb_{l}(\Ib_{mn} - \mathcal{H}_{\rm H})((\Ib_n-\Ab)^T\otimes \Ib_m)\zetab_{l} + \Cb_{l} \mathcal{Z}_{\rm H}\big) - \big(\Cb_{l}(\Ib_{mn} - \mathcal{H}_{\rm H}^\prime)((\Ib_n-\Ab)^T\otimes \Ib_m)\zetab_{l} + \Cb_{l} \mathcal{Z}_{\rm H}^\prime\big)\|\,\\
&\leq&\|\Cb_{l}(\mathcal{H}_{\rm H}^\prime - \mathcal{H}_{\rm H})((\Ib_n-\Ab)^T\otimes \Ib_m)\zetab_{l}\| + \|\Cb_{l} (\mathcal{Z}_{\rm H} - \mathcal{Z}_{\rm H}^\prime)\big)\|\,\\
&\leq&\mu\|((\Ib_n-\Ab)^T\otimes \Ib_m)\zetab_{l}\| + \mu\sqrt{n}\,\\
\ea\]
where the second inequality is obtained by using the facts that $\|\Cb_{l}\|_1=1$ and $\mathfrak{E}$ and $\mathfrak{E}^\prime$ are $\mu$-adjacent.
Thus, for $l=0$, we have
\[\ba{rcl}
&&\|\big(\Cb_{0}(\Ib_{mn} - \mathcal{H}_{\rm H})((\Ib_n-\Ab)^T\otimes \Ib_m)\zetab_{0} + \Cb_{0} \mathcal{Z}_{\rm H}\big) - \big(\Cb_{0}(\Ib_{mn} - \mathcal{H}_{\rm H}^\prime)((\Ib_n-\Ab)^T\otimes \Ib_m)\zetab_{0} + \Cb_{0} \mathcal{Z}_{\rm H}^\prime\big)\|\,\\
&\leq& \mu\sqrt{n}(\|\zeta_0\|+1)\,\\
&\leq& \mu(\sqrt{\nu}+\phi^\ast+\sqrt{n})\,,
\ea\]
and for $l=1,\ldots,L-1$, we have
\[\ba{rcl}
&&\|\big(\Cb_{l}(\Ib_{mn} - \mathcal{H}_{\rm H})((\Ib_n-\Ab)^T\otimes \Ib_m)\zetab_{l} + \Cb_{l} \mathcal{Z}_{\rm H}\big) - \big(\Cb_{l}(\Ib_{mn} - \mathcal{H}_{\rm H}^\prime)((\Ib_n-\Ab)^T\otimes \Ib_m)\zetab_{l} + \Cb_{l} \mathcal{Z}_{\rm H}^\prime\big)\|\,\\
&\leq& \mu\|\Delta_{l-1}(\phib_{l-1}) + \frac{1}{n}\big(\1_n\1_n^\top\otimes \Ib_m\big)\phib_{l-1}\| + \mu\sqrt{n}\,\\
&\leq& \mu(\sqrt{\nu}+\phi^\ast+\sqrt{n})
\ea\]
where the first inequality is obtained by using $\Delta_{l-1}(\phib_{l-1}) + \frac{1}{n}\big(\1_n\1_n^\top\otimes \Ib_m\big)\phib_{l-1} = ((\Ib_n-\Ab)^T\otimes \Ib_m)\zetab_{l}$, and the second is obtained by using the conditioned event $\bigcap_{l=0}^{L-1}\big(\mcalQ_S^{l} \bigcap \|\Delta_{l}(\phib_{l})) \|^2 \leq {\nu}\big)$ and Lemma \ref{lemma-E-nu}.

As in the arguments in Appendix \ref{app:sec-ppsc}, we can obtain
\[\ba{rcl}
&&\mathbb{P}(\mathscr{M}_{\rm eq}\big(\mathfrak{E} \,|\, \mathcal{E}_{\rm p}^\ast\big)\in R\big)= \mathbb{P}\big({\mathcal{F}_0^{\rm ea}}\big)\cdot\prod_{l=1}^{L-1}\mathbb{P}\big({\mathcal{F}_l^{\rm ea}}|\,\bigcap_{k=0}^{l-1}{\mathcal{F}_k^{\rm ea}}\big)\\
&\leq& \bigg(\exp({\epsilon\over L})\cdot\mathbb{P}\big({\mathcal{F}_0^{\rm ea}}^\prime\big) + \mathbb{P}\big(Z_0\geq \frac{\epsilon\sigma_{\gammab}\lambda_{\rm ppsc}}{L\mu(\sqrt{\nu}+\phi^\ast+\sqrt{n})}- \frac{\mu(\sqrt{\nu}+\phi^\ast+\sqrt{n})}{2\sigma_{\gammab}\lambda_{\rm ppsc}} \big)\bigg)\cdot\,\\
&&\cdot\prod_{l=1}^{L-1}\bigg(\exp({\epsilon\over L})\cdot\mathbb{P}\big({\mathcal{F}_l^{\rm ea}}^\prime|\,\bigcap_{k=1}^{l-1}{\mathcal{F}_k^{\rm ea}}^\prime\big) + \mathbb{P}\big(Z_l\geq \frac{\epsilon\sigma_{\gammab}\lambda_{\rm ppsc}}{L\mu(\sqrt{\nu}+\phi^\ast+\sqrt{n})}- \frac{\mu(\sqrt{\nu}+\phi^\ast+\sqrt{n})}{2\sigma_{\gammab}\lambda_{\rm ppsc}} \big)\bigg)\,\\
&\leq& \bigg(\exp({\epsilon\over L})\cdot\mathbb{P}\big({\mathcal{F}_0^{\rm ea}}^\prime\big)  + \delta_\ast\bigg)\cdot \prod_{l=1}^{L-1}\bigg(\exp({\epsilon\over L})\cdot\mathbb{P}\big({\mathcal{F}_l^{\rm ea}}^\prime|\,\bigcap_{k=1}^{l-1}{\mathcal{F}_k^{\rm ea}}^\prime\big)  + \delta_\ast\bigg)\,\\
&\leq& \exp(\epsilon)\mathbb{P}\big(\mathscr{M}_{\rm eq}\big(\mathfrak{E}^\prime \,|\, \mathcal{E}_{\rm p}^\ast\big)\in R\big) + \big(\exp\big(\frac{\epsilon}{L}\big)+\delta_\ast)^{L} - e^{\epsilon}
\ea\]
with $Z_l\sim \mathcal{N}(0,1)$, and
\[
\delta_\ast = \mathcal{Q}\bigg(\frac{\epsilon\sigma_{\gammab}\lambda_{\rm ppsc}}{L\mu(\sqrt{\nu}+\phi^\ast+\sqrt{n})}- \frac{\mu(\sqrt{\nu}+\phi^\ast+\sqrt{n})}{2\sigma_{\gammab}\lambda_{\rm ppsc}}\bigg)\,.
\]
With $\delta_{\sharp}=(\delta+e^{\epsilon})^{1\over L}-\exp(\frac{\epsilon}{L})$ and $\sigma_\gammab\geq\sigma_\gammab^\ast := {\mu\kappa(\frac{\epsilon}{L},\delta_{\sharp})(\sqrt{\nu}+\phi^\ast+\sqrt{n})}/ {\lambda_{\rm ppsc}}$,
we have
\[
\kappa(\frac{\epsilon}{L},\delta_{\sharp}) \leq \frac{\lambda_{\rm ppsc} \sigma_{\gammab}}{\mu(\sqrt{\nu}+\phi^\ast+\sqrt{n})}
\]
yielding $\delta_{\sharp} \geq \delta_{\ast}$.
This indicates $(\delta+e^{\epsilon})^{1\over L}\geq \delta_{\ast}+\exp(\frac{\epsilon}{L})$, leading to
\[
\delta \geq \big(\exp\big(\frac{\epsilon}{L}\big)+\delta_\ast\big)^{L} - e^{\epsilon}\,.
\]
Therefore,  there holds
\[
\mathbb{P}\big(\mathscr{M}_{\rm eq}\big(\mathfrak{E} \,|\, \mathcal{E}_{\rm p}^\ast\big)\in R\big)
\leq \exp(\epsilon)\mathbb{P}\big(\mathscr{M}_{\rm eq}\big(\mathfrak{E}^\prime \,|\, \mathcal{E}_{\rm p}^\ast\big)\in R\big) + \delta\,,
\]
which completes this part of proof.

%%%
\subsection{Proof of statement (ii)}

We observe that each $\Delta_{l}(\phib_{l})$ defined in (\ref{eq:Delta_l}), $l=0,1,\ldots,L-1$ can be regarded as the computation error of the PPSC-Gossip-AC algorithm with initial state $\phib_{l}$. Thus, by Lemma \ref{lemma-B}, we have
\[\ba{rcl}
\mathbb{E}\|\Delta_{l}(\phib_{l})\|^2 &\leq& (n\mathbb{E}\|\phib_{l}\|^2 + 2q^2S^2\sigma_{\gammab}^2) (1-\lambda_{\rm G})^{2T}\,\\
&\leq& \big(2n\mathbb{E}\|\phib_{l}- \1_n\otimes\yb^\ast\|^2 + 2n\|\yb^\ast\|^2+2q^2S^2\sigma_{\gammab}^2\big) (1-\lambda_{\rm G})^{2T}\,.
\ea\]
where the last inequality is obtained by using the H$\ddot{\textnormal{o}}$lder inequality.

In view of this, with (\ref{eq:phib-y}) and using the H$\ddot{\textnormal{o}}$lder inequality again, we thus have
\[\ba{l}
\mathbb{E}\|\phib_{l+1} - \1_n\otimes\yb^\ast\|^2
\leq \lambda_H^2\mathbb{E}\|\phib_{l} - \1_n\otimes\yb^\ast\|^2 + \mathbb{E}\|\Delta_{l}(\phib_{l})\|^2  + 2\sqrt{\lambda_H^2\mathbb{E}\|\phib_{l} - \1_n\otimes\yb^\ast\|^2\mathbb{E}\|\Delta_{l}(\phib_{l})\|^2} \\ \qquad
\leq (\lambda_H^2+2n(1-\lambda_{\rm G})^{2T})\mathbb{E}\|\phib_{l-1} - \1_n\otimes\yb^\ast\|^2 + \big[2n\|\yb^\ast\|^2+2q^2S^2\sigma_{\gammab}^2\big] (1-\lambda_{\rm G})^{2T}
\,\\ \qquad + \bigg[3n\mathbb{E}\|\phib_{l}-\1_n\otimes\yb^\ast\|^2  + 2n\|\yb^\ast\|^2  +2q^2S^2\sigma_{\gammab}^2\bigg] \cdot (1-\lambda_{\rm G})^T\,\\
\leq \big(\lambda_H^2+5n(1-\lambda_{\rm G})^{T}\big)\mathbb{E}\|\phib_{l} - \1_n\otimes\yb^\ast\|^2  + 4\big(n\|\yb^\ast\|^2 +q^2S^2\sigma_{\gammab}^2\big) (1-\lambda_{\rm G})^{T}\,.\\
\ea\]
This yields
\[\ba{l}
\mathbb{E}\|\phib_L - \1\otimes\yb^\ast\|^2 \leq \|\phib_0-\1_n\otimes\yb^\ast\|^2 (\lambda_H^2+5n(1-\lambda_{\rm G})^{T})^{L}  + \frac{4( n\|\yb^\ast\|^2+S^2\sigma_{\gammab}^2)(1-\lambda_{\rm G})^{T}}{1-\lambda_H^2-5n(1-\lambda_{\rm G})^{T}}\,
\ea\]
with $\|\phib_0-\1_n\otimes\yb^\ast\|^2\leq n\|\zeta_0-\yb^\ast\|^2$.
With the selections of $L$ and $T$ in Theorem \ref{theorem-Equ}, we have
\[\ba{l}
\lambda_H^2+5n(1-\lambda_{\rm G})^{T} \leq \frac{\varepsilon_0+\lambda_H^2}{1+\varepsilon_0}\,\\
\|\phib_0-\1_n\otimes\yb^\ast\|^2 (\lambda_H^2+5n(1-\lambda_{\rm G})^{T})^{L} \leq \frac{\nu}{2}\,\\
\frac{4( n\|\yb^\ast\|^2+S^2\sigma_{\gammab}^2)(1-\lambda_{\rm G})^{T}}{1-\lambda_H^2-5n(1-\lambda_{\rm G})^{T}} \leq \frac{\nu}{2}\,.
\ea\]
This in turn completes the proof.

%%%%%%%%%%%%%%%%%%%%%%%%%%%%%%%%%%%%%%%%%%%%%%%%%%%%%
%%%%%%%%%%%%%%%%%%%%%%%%%%%%%%%%%%%%%%%%%%%%%%%%%%%%%%%
\section{Proof of Theorem \ref{theorem-Op}}
\label{app:sec-opt}
%%
%%%
\subsection{Preliminaries}

Following the terminologies in Appendix \ref{app:equ}, we let $\zetab_{0}= \1_n\otimes \zeta_0$ and $\zetab_l\in\R^{nm}$ denote the netowrk state vector after the $l$ PPSC-Gossip procedure, and $\phib_0=\xb_0$ and $\phib_l$ denote the network state vector after the $l$-th procedure of projected subgradient descent. Denote
\[
  \mathpzc{P}_{\mathtt{C}^n}(x_1,\ldots,x_n)= \begin{pmatrix}
                                       \mathpzc{P}_{\mathtt{C}}(x_1) \\
                                       \vdots \\
                                       \mathpzc{P}_{\mathtt{C}}(x_n)
                                     \end{pmatrix}
                                      \,,\quad \mbox{for $x_k\in\R^m$}\,
\]
  and
  \[
  \nabla \Fb(x_1,\ldots,x_n) = \begin{pmatrix}
                               \nabla f_1(x_1) \\
                               \vdots \\
                               \nabla f_n(x_n)
                             \end{pmatrix}
                             = \begin{pmatrix}
                               \frac{\partial f_1}{\partial x_1}(x_1) \\
                               \vdots \\
                               \frac{\partial f_n}{\partial x_n} f_n(x_n)
                             \end{pmatrix}
                             \,,\quad \mbox{for $x_k\in\R^m$}\,.
  \]
According to Algorithm 5, we have
\beeq{\label{eq:AC-zeta}\ba{rl}
  \zetab_{l+1} &= \mathsf{C}_{\mathcal{E}_{{\rm p},l}} \phib_{l} + \mathsf{D}_{\mathcal{E}_{{\rm p},l}} \gammab_{l}\,,\quad \mbox{for $l=0,1,\ldots,L-1$}
  \ea
}
and
\beeq{\label{eq:AC-phi}\ba{rcl}
  \phib_{l+1} &=& \mathpzc{P}_{\mathtt{C}^n}\left(\big((\Ib_n-\Ab)^T\otimes \Ib_m\big)\zetab_{l+1}-\alpha_{l+1}\nabla \Fb\big(\big((\Ib_n-\Ab)^T\otimes \Ib_m\big)\zetab_{l+1}\big)\right)\,\,,\quad \mbox{for $l=0,1,\ldots,L-1$} \\
%  &=& \mathpzc{P}_{X^n}\bigg((\Ab^T\otimes \Ib)(\mathsf{C}_{\mathcal{E}_{{\rm p},l}} \phib_{l} + \mathsf{D}_{\mathcal{E}_{{\rm p},l}}\Gamma_{l})-\alpha_{l}\nabla \Fb\big((\Ab^T\otimes \Ib)  (\mathsf{C}_{\mathcal{E}_{{\rm p},l}} \phib_{l} + \mathsf{D}_{\mathcal{E}_{{\rm p},l}}\Gamma_{l})\big)\bigg)
  \ea
  }

\subsection{Proof of statement (i)}

We define
\[\ba{rcl}
\psib_l(\phib_{l}) &=& \mathpzc{P}_{\mathtt{C}^n}\left(\big((\Ib_n-\Ab)^T\otimes \Ib_m\big)\zetab_{l+1}-\alpha_{l+1}\nabla \Fb\big(\big((\Ib_n-\Ab)^T\otimes \Ib_m\big)\zetab_{l+1}\big)\right) \\ && - \big((\Ib_n-\Ab)^T\otimes \Ib_m\big)\zetab_{l+1}+\alpha_{l+1}\nabla \Fb\big(\big((\Ib_n-\Ab)^T\otimes \Ib_m\big)\zetab_{l+1}\big)\,\\
\Delta_{l}(\phib_{l}) &=& ((\Ib_n-\Ab)^T\otimes \Ib_m)[\mathsf{C}_{\mathcal{E}_{{\rm p},l}} \phib_{l} + \mathsf{D}_{\mathcal{E}_{{\rm p},l}}\Gamma_{l}] - \frac{1}{n}[(\1_n\1_n^\top)\otimes \Ib_m] \phib_{l}\,.
\ea\]
This then rewrites (\ref{eq:AC-phi})  as
\[\ba{l}
\phib_{l+1}
= \frac{1}{n}(\1_n\1_n^\top\otimes \Ib_m) \phib_{l}  + \Delta_{l}(\phib_{l})  -\alpha_{l+1}\nabla \Fb\big(\frac{1}{n}(\1_n\1_n^\top\otimes \Ib_m) \phib_{l} + \Delta_{l}(\phib_{l})\big)  +\psib_l(\phib_{l})\,.
\ea
\]

To complete this part of proof,  we present the following lemma.
\begin{lemma}\label{lemma-F}
Suppose $\|\Delta_{l}(\phib_{l})\| \leq \sqrt{\nu}\alpha_{l+1}^2$ for all $l\geq 0$. Then there exists an $L^\dag$, independent of $S$, $T$ and $\sigma_\gammab$, and a $\yb^\dag\in \mathtt{C}^\dag$ such that $\|\phib_L - \1_n\otimes\yb^\dag\|^2\leq \nu$ holds for all $L\geq L^\dag$.
\end{lemma}
%%%
\begin{proof}
We denote $\theta_l=\frac{1}{n}(\1_n^\top\otimes \Ib_m)\phib_{l}- \yb^\ddag$ for any $\yb^\ddag\in \mathtt{C}^\dag$, which yields
\[
\theta_{l+1} = \theta_l + \frac{1}{n}(\1_n^\top\otimes\Ib_m)\Delta_l - \frac{\alpha_{l+1}}{n}(\1_n^\top\otimes\Ib_m)\nabla \Fb\big(\1_n\otimes(\theta_l+\yb^\ddag) + \Delta_l \big) + \frac{1}{n}(\1_n^\top\otimes \Ib_m)\psib_l(\phib_{l})\,.
\]
According to Lemma 1.(b) in \cite{nedic2010constrained}, this further implies
\[\ba{rcl}
\|\theta_{l+1}\|^2
&\leq& \left\|\theta_l + \frac{1}{n}(\1_n^\top\otimes\Ib_m)\Delta_l - \frac{\alpha_{l+1}}{n}(\1_n^\top\otimes\Ib_m)\nabla \Fb\big(\1_n\otimes(\theta_l+\yb^\ddag) + \Delta_l \big)\right\|^2  - \frac{1}{n}\|\psib_l(\phib_{l})\|^2\,\\
&\leq&  \|\theta_{l}\|^2 + \frac{1}{n}\|\Delta_{l}\|^2 + \frac{2}{n}\| \Delta_{l}\| \|\theta_{l}\| +\frac{\alpha_{l+1}^2}{n}\left\|\nabla \Fb\big(\1_n\otimes(\theta_l+\yb^\ddag) + \Delta_l \big)\right\|^2 \,\\ && - \frac{2\alpha_{l+1}}{n}\1_n^\top\left[\Fb\big(\1_n\otimes(\theta_l+\yb^\dag) + \Delta_l\big) - \Fb(\1_n\otimes\yb^\ddag)\right] - \frac{1}{n} \|\psib_l(\phib_{l})\|^2\,.
\ea
\]
If $\|\Delta_{l}(\phib_{l})\|\leq \sqrt{\nu}\alpha_{l+1}^2$ for all $l=1,2,\ldots,L$, we then have
\beeq{\label{eq:theta_l+1}\ba{rcl}
\|\theta_{l+1}\|^2
&\leq&   \|\theta_{l}\|^2 + \frac{\nu}{n}\alpha_{l+1}^4 + \frac{2\sqrt{\nu}}{n}\alpha_{l+1}^2 \|\theta_{l}\| +\frac{\alpha_{l+1}^2}{n}\left\|\nabla \Fb\big(\1_n\otimes(\theta_l+\yb^\ddag) + \Delta_l \big)\right\|^2 \,\\ && + \sqrt{\nu}\alpha_{l+1}^3 \left\|\nabla \Fb\big(\1_n\otimes(\theta_l+\yb^\ddag) + \Delta_l \big)\right\| \\ && - \frac{2\alpha_{l+1}}{n}\1_n^\top\left[\Fb\big(\1_n\otimes(\theta_l+\yb^\ddag)\big) - \Fb(\1_n\otimes\yb^\ddag)\right] - \frac{1}{n} \|\psib_l(\phib_{l})\|^2\,
\ea
}
for each $\yb^\ddag\in \mathtt{C}^\dag$.
Note that each $\nabla f_k$ is bounded over the set $\Omega_\nu$, i.e., there exists a $L_{\nu}>0$ such that $\|\nabla f_k(\yb)\|\leq L_\nu$, for all $\yb\in\Omega_{\nu}$ and $k\in{\rm V}$.
Since $\theta_l+\yb^\ddag=\frac{1}{n}(\1_n^\top\otimes \Ib_m)\phib_{l}\in \mathtt{C}\subset\Omega_\nu$, we have
\beeq{\label{eq:appE-1}\ba{l}
\left\|\nabla \Fb\big(\1_n\otimes(\theta_l+\yb^\ddag)  \big)\right\|^2 \leq n L_{\nu}^2\,\\
\eb_i^\top\left[\Fb\big(\1_n\otimes(\theta_l+\yb^\ddag) + \Delta_l\big) - \Fb(\1_n\otimes\yb^\ddag)\right] \geq 0\,,\quad \mbox{for all $i=1,\ldots,n$}
\ea}
which yields
\begin{equation}\label{eq:33}
\ba{rcl}
\|\theta_{l+1}\|^2
&\leq&   \|\theta_{l}\|^2 + \frac{\nu}{n}\alpha_{l+1}^4 + \frac{2\sqrt{\nu}}{n}\alpha_{l+1}^2 \|\theta_{l}\| +\alpha_{l+1}^2L_{\nu}^2 + \sqrt{n\nu}\alpha_{l+1}^3 L_{\nu}   - \frac{1}{n} \|\psib_l(\phib_{l})\|^2\,
\ea
\end{equation}
for each $\yb^\ddag\in \mathtt{C}^\dag$.
By dropping the last negative term and using the facts that $\|\theta_{l-1}\|$ is bounded by $\theta_l+\yb^\ddag=\frac{1}{n}\sum_{k=1}^n \phib_{l,k}\in \mathtt{C}$ and $\lim_{l\rightarrow\infty}\sum_{i=1}^l\alpha_i^2<\infty$, we can conclude that the sequence $\{\|\theta_l\|\}$, i.e., $\{\|\frac{1}{n}(\1_n^\top\otimes \Ib_m)\phib_{l}- \yb^\ddag\|\}$ is convergent for each $\yb^\ddag\in \mathtt{C}^\dag$. Since $\frac{1}{n}(\1_n^\top\otimes \Ib_m)\phib_{l}$ is bounded, it must have a limit point $\yb^\dag$. On the other hand, by (\ref{eq:theta_l+1}), we can obtain
\[\ba{l}
\|\theta_{l+1}\|^2 + \sum\limits_{i=0}^l\frac{2\alpha_{i+1}}{n}\1_n^\top\left[\Fb\big(\frac{1}{n}(\1_n\1_n^\top\otimes \Ib_m)\phib_{i} \big) - \Fb(\1_n\otimes\yb^\dag)\right] \\ \qquad \leq
 \|\theta_{0}\|^2   + \sum\limits_{i=0}^l\big(\frac{\nu}{n}\alpha_{i+1}^4 + \frac{2\sqrt{\nu}}{n}\alpha_{i+1}^2 \|\theta_{i}\| +\alpha_{i+1}^2L_{\nu}^2 + \sqrt{n\nu}\alpha_{i+1}^3 L_{\nu}\big)\,,
\ea\]
the right side of which is bounded  by $\lim_{l\rightarrow\infty}\sum_{i=1}^l\alpha_{i}^2<\infty$ and the fact that $\|\theta_{i}\|$ is bounded. This in turn implies
\[
\lim_{l\rightarrow\infty}\sum\limits_{i=0}^l\frac{2\alpha_{i+1}}{n}\1_n^\top\left[\Fb\big(\frac{1}{n}(\1_n\1_n^\top\otimes \Ib_m)\phib_{i} \big) - \Fb(\1_n\otimes\yb^\dag)\right]<\infty\,.
\]
By $\lim_{l\rightarrow\infty}\sum_{i=1}^l \alpha_l=\infty$ and the latter of (\ref{eq:appE-1}), we have
\[
\lim_{l\rightarrow\infty}\Fb\big(\frac{1}{n}(\1_n\1_n^\top\otimes \Ib_m)\phib_{l} \big) - \Fb(\1_n\otimes\yb^\dag)=\Fb\big(\1_n\otimes\yb^\dag \big) - \Fb(\1_n\otimes\yb^\ddag) = 0\,,
\]
which yields $\yb^\dag\in \mathtt{C}^\dag$. Therefore, by fixing $\yb^\ddag$ as $\yb^\dag$ in the definition of $\theta_l$, we have $\lim\limits_{l\rightarrow\infty}\theta_l = 0$.

Bearing in mind the previous analysis, we then observe that
\[\ba{l}
\phib_{l+1} - \1_n\otimes \yb^\dag
= \1_n\otimes \theta_{l}  + \Delta_{l}(\phib_{l})  -\alpha_{l+1}\nabla \Fb\big(\1_n\otimes(\theta_l+\yb^\dag) + \Delta_{l}(\phib_{l})\big)  +\psib_l(\phib_{l})
\ea
\]
which, if $\|\Delta_{l}(\phib_{l})\| \leq \sqrt{\nu}\alpha_{l+1}^2$ for all $l=0,1,\ldots,L-1$, leads to
\[
\|\phib_{l+1} - \1_n\otimes\yb^\dag\| \leq \sqrt{n}\|\theta_{l}\| + \sqrt{\nu}\alpha_{l+1}^2 + \alpha_{l+1}nL_{\nu} + \|\psib_l(\phib_{l})\| \,.
\]
Regarding the bound of the last term in the above inequality, we observe that
\[\ba{rcl}
\psib_{l}(\phib_{l}) &=&  \mathpzc{P}_{\mathtt{C}^n}\left(\1_n\otimes(\theta_l+\yb^\dag) + \Delta_{l}(\phib_{l})-\alpha_{l+1}\nabla \Fb\big(\1_n\otimes(\theta_l+\yb^\dag) + \Delta_{l}(\phib_{l})\big)\right) \\ && -\left(\1_n\otimes(\theta_l+\yb^\dag) + \Delta_{l}(\phib_{l})-\alpha_{l+1}\nabla \Fb\big(\1_n\otimes(\theta_l+\yb^\dag) + \Delta_{l}(\phib_{l})\big)\right)
\ea\]
which immediately follows that there exist $c_1,c_2>0$ such that
$\|\psib_{l}(\phib_{l})\| \leq c_1\|\theta_{l}\| +  c_2\alpha_{l+1}$.
Thus, as $\lim_{l\rightarrow\infty}\|\theta_l\| = 0$ and $\lim_{l\rightarrow\infty}\alpha_l = 0$, we have $\lim_{l\rightarrow\infty}\|\phib_l - \1_n\otimes\yb^\dag\| = 0$, which completes the proof.
\end{proof}

By Lemma \ref{lemma-B}, we have $\mathbb{E}\left(\|\Delta_{l}(\phib_{l})\|^2\right) \geq (n\|\phib_{l}\|^2 + 2q^2S^2\sigma_{\gammab}^2)  (1-\lambda_{\rm G})^{2T}$ for all $l\geq0$.
Using Markov's inequality, this implies that for any $\nu>0$,
\[
\mathbb{P}\big(\|\Delta_{l}(\phib_{l})\|^2 \geq \nu\alpha_{l+1}^4\big) \leq \frac{\mathbb{E}\left(\|\Delta_{l}(\phib_{l})\|^2\right)}{\nu\alpha_{l+1}^4}\,
\]
yielding
\[
\mathbb{P}\left(\bigcap_{l=0}^{L-1}\|\Delta_{l}(\phib_{l})) \| \leq \sqrt{\nu} \alpha_{l+1}^2\right) \geq \prod_{l=0}^{L-1}\bigg(1- \frac{[ n\|\phib_{l}\|^2 + 2q^2S^2\sigma_{\gammab}^2](1-\lambda_{\rm G})^{2T}}{\nu\alpha_{l+1}^4} \bigg)\,.
\]
Thus,  with $\phib_l\in \mathtt{C}$ for $l\geq 0$ and $T\geq \frac{\log((1-p^{1\over L})\nu\alpha_L^4)-\log(n{\phi^\dag}^2 + 2q^2S^2\sigma_{\gammab}^2)}{2\log(1-\lambda_{\rm G})}$, we have
\beeq{\label{eq:appF-1}
\mathbb{P}\left(\bigcap_{l=0}^{L-1}\|\Delta_{l}(\phib_{l})) \| \leq \sqrt{\nu} \alpha_{l+1}^2\right) \geq p\,.
}
Therefore, the proof is completed by combining Lemma \ref{lemma-F} and (\ref{eq:appF-1}).

\subsection{Proof of statement (ii)}

By the arguments of Lemma 1, it immediately follows that the event $\mcalQ_S^l$ that all node states have altered during the time $s\in[l(S+T+1)+1,l(S+T+1)+S]$ occurs with probability larger than $\rho^{1\over L}$, i.e., $\mathbb{P}(\mcalQ_S^l)\geq \rho^{1\over L}$ and $\mathbb{P}\left(\bigcap_{l=0}^{L-1}\mcalQ_S^l\right) \geq \rho$.
With this in mind, we now proceed to analyze the differential privacy of PPSC-Gossip-DCO algorithm, conditioned on the event $\bigcap_{l=0}^{L-1}\mcalQ_S^l$.

As in Appendix \ref{app:equ}, we denote $\mathcal{E}_{\rm p}^\ast:=\big\{\mathcal{E}_{{\rm p},0}^\ast,\mathcal{E}_{{\rm p},1}^\ast,\ldots,\mathcal{E}_{{\rm p},L-1}^\ast\big\}$ as the sequence of communication edges that is observed by the eavesdroppers during $L$ recursions of the Multi-Gossiping PPSC mechanism, and ${\Cb}_{l}\in\R^{mn\times mn}$ and $\Db_{l}\in\R^{mn\times mS}$ be the resulting values of $\mathsf{C}_{\mathcal{E}_{{\rm p},l}^\ast}$ and $\mathsf{C}_{\mathcal{E}_{{\rm p},l}^\ast}$.

%In view of the previous analysis, we have
%\[\ba{l}
%\mathbb{P}\big(\mathscr{M}_{\rm op}(\mathcal{F} \,|\, \mathcal{E}_{\rm p}^\ast)\in R\big)\\ \overset{}{=} \mathbb{P}\bigg(\bigcap\limits_{l=1}^L U_{l}^\top\Cb_{l} \mathpzc{P}_{X^n}\left((\Ab^T\otimes \Ib)\zetab_{l-1}-\alpha_{l-1}\nabla \Fb\big((\Ab^T\otimes \Ib)\zetab_{l-1}\big)\right) \Sigma_{l}\Gamma_{l}\in R_i'\bigg)\,\\
%\overset{}{=}\prod\limits_{l=1}^L\frac{(2\pi)^{-\frac{S}{2}}}{\sigma_{\gammab}\det(\widehat\Sigma)^{\frac{1}{2}}}\int_{R_l'}\exp\bigg(-\frac{1}{2\sigma_{\gammab}^2}\bigg(u-U_{l}^\top\Cb_{l} \mathpzc{P}_{X^n}\big((\Ab^T\otimes \Ib)\zetab_{l-1}\,\\ \qquad \qquad \qquad -\alpha_{l-1}\nabla \Fb\big((\Ab^T\otimes \Ib)\zetab_{l-1}\big)\big)\bigg)^\top\widehat\Sigma_l^{-1} \\ \qquad \cdot \bigg(u-U_{l}^\top\Cb_{l} \mathpzc{P}_{X^n}\left((\Ab^T\otimes \Ib)\zetab_{l-1}-\alpha_{l-1}\nabla \Fb\big((\Ab^T\otimes \Ib)\zetab_{l-1}\big)\right)\bigg){\rm d} u\,\\
%%%
%%}{\leq} \prod_{l=1}^L\left(\exp({\epsilon\over L})\cdot\mathbb{P}(\mathscr{M}_a\big({\bf d}^\prime \,|\, \mathcal{E}_{\rm p}^\ast)\in R\big)  + {\color{blue} ???} \mathbb{P}\bigg( \frac{1}{\sigma_{\gammab}}{\tilde{\bf d}}^\top \Cb^\top U_a \Sigma^{-\frac{1}{2}}v \geq \epsilon- \frac{1}{2\sigma_{\gammab}^2}\|\Sigma^{-\frac{1}{2}}U_a^\top\Cb\tilde{\bf d}\|^2\bigg)\right)\,
%\ea\]	
%where $\widehat\Sigma_l=\Sigma_l\Sigma_l^\top$.

Let $R_l\subseteq\mathbb{R}^{mn}$ and define the events
\[\ba{l}
{\mathcal{F}_l^{op}}:=\left\{\Cb_{l} \mathpzc{P}_{\mathtt{C}^n}\big(((\Ib_n-\Ab)^T\otimes \Ib_m)\zetab_{l}-\alpha_{l}\nabla \Fb\big(((\Ib_n-\Ab)^T\otimes \Ib_m)\zetab_{l}\big)\big) + \Db_{l}\gammab_{l}\in R_i\right\}\\
{\mathcal{F}_l^{op}}^\prime:=\left\{\Cb_{l} \mathpzc{P}_{\mathtt{C}^n}\big(((\Ib_n-\Ab)^T\otimes \Ib_m)\zetab_{l}-\alpha_{l}\nabla \Fb^\prime\big(((\Ib_n-\Ab)^T\otimes \Ib_m)\zetab_{l}\big)\big) + \Db_{l}\gammab_{l}\in R_i\right\}
\ea\]
for $l=0,1,\ldots,L-1$. Further, we define
\[\ba{l}
\widetilde{\mathpzc{P}}_{l} :=  \mathpzc{P}_{\mathtt{C}^n}\bigg(((\Ib_n-\Ab)^T\otimes \Ib_m)\zetab_{l}-\alpha_{l}\nabla \Fb\big(((\Ib_n-\Ab)^T\otimes \Ib_m)\zetab_{l}\big)\bigg)\,\\ \qquad \quad -\mathpzc{P}_{\mathtt{C}^n}\bigg(((\Ib_n-\Ab)^T\otimes \Ib_m)\zetab_{l}-\alpha_{l}\nabla \Fb^\prime\big(((\Ib_n-\Ab)^T\otimes \Ib_m)\zetab_{l}\big)\bigg)\,,
\ea\]
which for any $\mu$-adjacent $\mathbf{F}^\prime, \mathbf{F}$, satisfies
\[\ba{rcl}
\|\widetilde{\mathpzc{P}}_{l}\| &\leq& \alpha_{l}\|\nabla \Fb\big(((\Ib_n-\Ab)^T\otimes \Ib)\zetab_{l}\big)-\nabla \Fb^\prime\big(((\Ib_n-\Ab)^T\otimes \Ib)\zetab_{l}\big)\|\,\\
&\leq& \alpha_{l}\sum_{i=1}^{n}\|(\tau_i-\tau_i^\prime)\| g^\dag\,\\
&\leq& ng^\dag\mu\,.
\ea\]

Bearing in mind the previous analysis, as in the arguments in Appendices \ref{app:sec-ppsc} and \ref{app:equ}, we can obtain
\[\ba{rcl}
&&\mathbb{P}(\mathscr{M}_{\rm op}\big(\mathbf{F} \,|\, \mathcal{E}_{\rm p}^\ast\big)\in R\big)\,\\
&\leq& \bigg(\exp({\epsilon\over L})\cdot\mathbb{P}\big({\mathcal{F}_0^{op}}^\prime\big)  + \mathbb{P}\bigg(Z_l\geq \frac{\epsilon\sigma_{\gammab}\lambda_{\rm ppsc}}{Lng^\dag\mu}- \frac{ng^\dag\mu}{2\sigma_{\gammab}}\lambda_{\rm ppsc} \bigg)\bigg)\,\\
&& \cdot \prod\limits_{l=1}^{L-1}\bigg(\exp({\epsilon\over L})\cdot\mathbb{P}\big({\mathcal{F}_l^{op}}^\prime|\,\bigcap\limits_{k=1}^{l-1}{\mathcal{F}_k^{op}}^\prime\big)  + \mathbb{P}\bigg(Z_l\geq \frac{\epsilon\sigma_{\gammab}\lambda_{\rm ppsc}}{Lng^\dag\mu}- \frac{ng^\dag\mu}{2\sigma_{\gammab}\lambda_{\rm ppsc}} \bigg)\bigg)\,\\
&\leq& \bigg(\exp({\epsilon\over L})\cdot\mathbb{P}\big({\mathcal{F}_0^{op}}^\prime\big)  + \delta_\dag\bigg)\cdot \prod_{l=1}^{L-1}\bigg(\exp({\epsilon\over L})\cdot\mathbb{P}\big({\mathcal{F}_l^{op}}^\prime|\,\bigcap_{k=1}^{l-1}{\mathcal{F}_k^{op}}^\prime\big)  + \delta_\dag\bigg)\,\\
&\leq& \exp(\epsilon)\mathbb{P}\big(\mathscr{M}_{\rm op}(\mathbf{F}^\prime\,|\, \mathcal{E}_{\rm p}^\ast)\in R\big) + \big(\exp({\epsilon\over L})+\delta_\dag\big)^{L}-e^{\epsilon}
\ea\]
with  $Z_l\sim \mathcal{N}(0,1)$, and $
\delta_\dag = \mathcal{Q}\bigg(\frac{\epsilon\sigma_{\gammab}\lambda_{\rm ppsc}}{Lng^\dag\mu }- \frac{ng^\dag\mu}{2\sigma_{\gammab}\lambda_{\rm ppsc}}\bigg)$.
With $\sigma_{\gammab}\geq \frac{n\mu g^\dag\mathcal{R}(\frac{\epsilon}{L},\delta_\sharp)} {\lambda_{\rm ppsc}}$ and $\delta_\sharp=(\delta+e^{\epsilon})^{1\over L}-\exp({\epsilon\over L})$, it can be seen that $\big(\exp({\epsilon\over L})+\delta_\dag\big)^{L}-e^{\epsilon}\leq \delta$.
The proof is thus completed.

%\bibliographystyle{IEEEtran}
%\bibliography{reference}

% that's all folks
\end{document}